\documentclass[11pt]{article}
\usepackage{amsfonts}
\usepackage{natbib}
\usepackage{amsmath}
\usepackage{mathrsfs}
\usepackage{amssymb}
\usepackage{amsmath,color}
\usepackage{amsthm}
\usepackage{amstext}
\usepackage{amsopn}
\usepackage[flushleft]{threeparttable}
\usepackage{texdraw}
\usepackage{graphicx}
\usepackage{multirow}
\usepackage{lscape}
\usepackage{epstopdf}
\usepackage{float}

\usepackage{subfigure}
\newtheorem{theorem}{Theorem}[section]

\newtheorem{corollary}[theorem]{Corollary}
\newtheorem{proposition}[theorem]{Proposition}

\DeclareMathOperator{\Tr}{Tr}
\DeclareMathOperator{\Det}{Det}

\theoremstyle{remark}

\numberwithin{equation}{section}

\begin{document}
\date{}
\title{\large \textbf{{Model of pattern formation in marsh ecosystems with nonlocal interactions}\footnote{Partially supported by NSF grant DMS-1716445 and DMS-1313093.} }}
\author
{ Sofya Zaytseva\textsuperscript{a}\footnote{Corresponding author. Email: \texttt{szaytseva@email.wm.edu}} \ \ Junping Shi\textsuperscript{b} \ \  Leah B Shaw\textsuperscript{b} \\
\\
{\small \textsuperscript{a} Department of Applied Science, College of William and Mary, \hfill{\ }}\\
\ \ {\small Williamsburg, Virginia 23187-8795, USA \hfill {\ }}\\
{\small \textsuperscript{b} Department of Mathematics, College of William and Mary, \hfill{\ }}\\
\ \ {\small Williamsburg, Virginia 23187-8795, USA \hfill{\ }}
}
\maketitle


\newcommand{\Om}{\Omega}
\newcommand{\ep}{\varepsilon}
\newcommand{\la}{\lambda}
\newcommand{\R}{{\mathbb R}}
\newcommand{\al}{\alpha}
\newcommand{\io}{\int_{\Om}}
\newcommand{\ds}{\displaystyle}
\newcommand{\noi}{\noindent}
\newcommand{\N}{\mathbb{N}}

\begin{abstract}
{Smooth cordgrass \textit{Spartina alterniflora} is a grass species commonly found in tidal marshes. It is an ecosystem engineer, capable of modifying the structure of its surrounding environment through various feedbacks. The scale-dependent feedback between marsh grass and sediment volume is particularly of interest. Locally, the marsh vegetation attenuates hydrodynamic energy, enhancing sediment accretion and promoting further vegetation growth. In turn, the diverted water flow promotes the formation of erosion troughs over longer distances. This scale-dependent feedback may explain the characteristic spatially varying marsh shoreline, commonly observed in nature. We propose a mathematical framework to model grass-sediment dynamics as a system of reaction-diffusion equations with an additional nonlocal term quantifying the short-range positive and long-range negative grass-sediment interactions. We use a Mexican-hat kernel function to model this scale-dependent feedback. We perform a steady state biharmonic approximation of our system and derive conditions for the emergence of spatial patterns, corresponding to a spatially varying marsh shoreline. We find that the emergence of such patterns depends on the spatial scale and strength of the scale-dependent feedback, specified by the width and amplitude of the Mexican-hat kernel function.}
\end{abstract}
\vspace{0.1in}


\vspace{0.1in}


\section{Introduction}

\label{intro}
Tidal marshes are among the richest and most productive ecosystems, supporting a variety of wildlife, serving as storm and erosion buffers, and playing an important role in improving water quality \citep{perry2009york,fagherazzi2013marsh,fagherazzi2014coastal}. The global loss of these ecosystems in the recent decades has motivated much research to understand their dynamics and aid in their restoration and management \citep{Deegan12,priestas2015coupled}. Marsh evolution is dynamic and complex, combining various biological and morphological processes happening not only in the marsh itself, but also in the tidal flat that borders it. As a result of these forces and interactions, a sharp scarp separating the marsh and tidal flat becomes a characteristic feature. The processes that take place on this scarp (i.e., marsh edge) influence whether the marsh recedes or expands \citep{tonelli2010modeling}. Various configurations of the marsh edge can be observed in nature, ranging from a uniform to a more jaggedy, sinusoidal shoreline (Figure \ref{fig:1}). While previous ecogeomorphic models have carefully considered the effects of sea-level rise, marsh vegetation colonization, wave activity, sediment fluxes, and underlying hydrodynamics \citep{mariotti2010numerical,tonelli2010modeling,fagherazzi2012numerical,schile2014modeling}, most of these have been numerical, computationally intensive models. We propose a simpler, phenomenological model to describe the large-scale evolution of the marsh in the horizontal direction in terms of two-way interactions between marsh vegetation and sedimentation. In particular, we are interested in the scale-dependent feedback present between marsh vegetation and sedimentation and whether this scale-dependent feedback can explain the spatially varying marsh shoreline, observed in nature.

 \begin{figure}[t]
	\centering
	\includegraphics[width=0.5\textwidth]{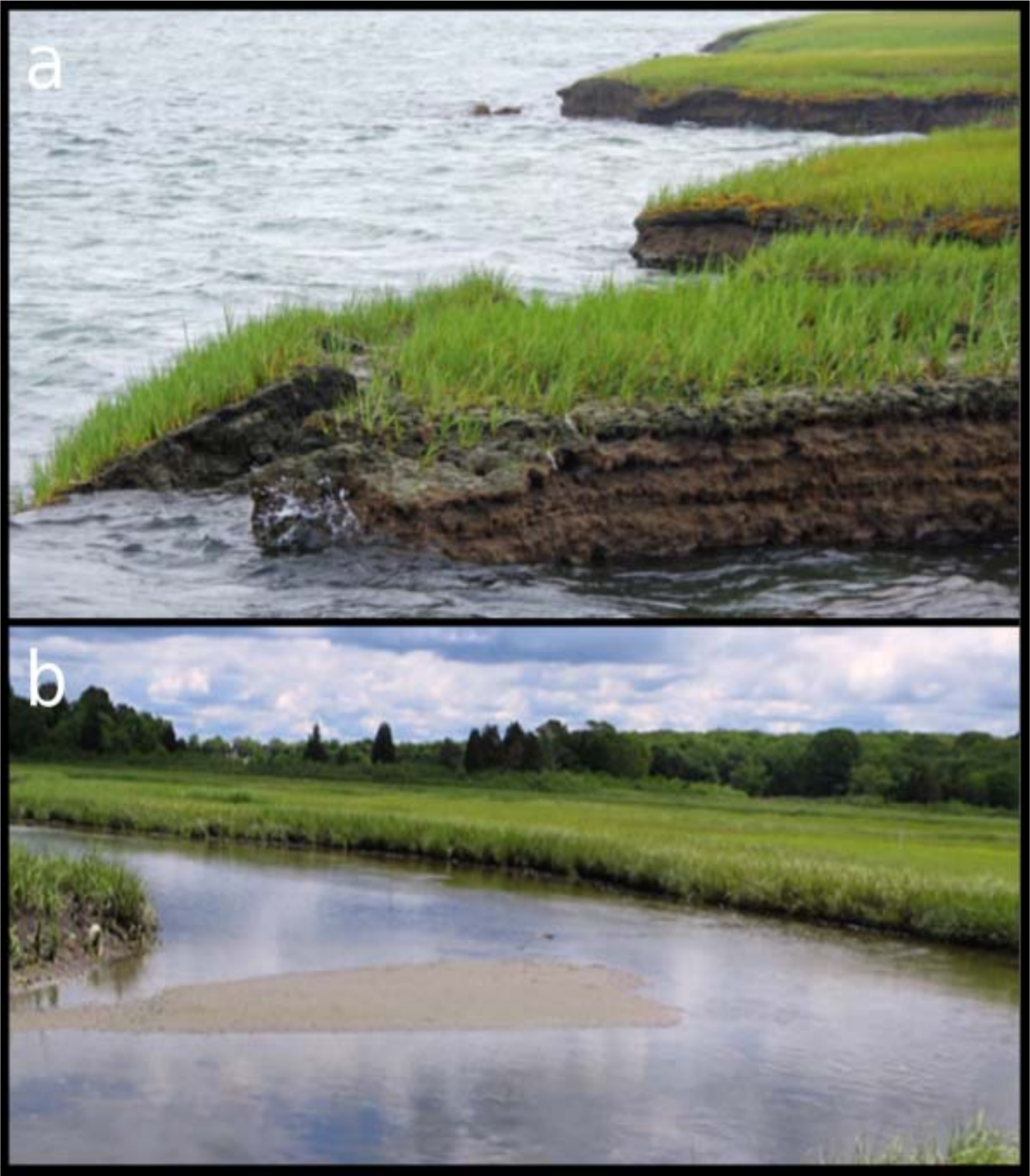}
	\caption{a) Self-organization on the marsh edge in the York River, a tributary of the Chesapeake Bay. Credit: Rom Lipcius. b) Uniform marsh edge. Credit: NOAA National Ocean Service }  \label{fig:1}
\end{figure}

Scale-dependent feedbacks are characterized by the presence of positive and negative interactions that happen at different spatial scales. In particular, scale-dependent feedbacks involving long-range negative interactions and short-range positive interactions are thought to be crucial for pattern development \citep{gierer1972theory, green2015positional, hiscock2015mathematically}, explaining spatially varying patterns in chemical \citep{castets1990experimental,rovinsky1993self}, biological \citep{nakamasu2009interactions,raspopovic2014digit} and ecological systems \citep{Rietkerk08}. In ecological systems, such scale-dependent feedbacks are thought to explain patterns in arid savannas, mussel beds, coral and oyster reefs, mudflats and other ecosystems \citep{Rietkerk08,van2012ecosystem,dibner2015ecological, demussels,pringle2017spatial,Barbier08}. Previously, we proposed a mathematical framework to investigate the evolution of the marsh edge as a result of scale-dependent interactions between sedimentation dynamics and two common marsh species, ribbed mussels (\textit{Geukensia demissa}) and smooth cordgrass (\textit{Spartina alterniflora}) \citep{zaytseva2018}, whose facilitatory nature and positive feedbacks have a significant effect on marsh development and proliferation \citep{Bertness84,Bertness85,watt2010population, altieri2007hierarchical}. While mussels are commonly found in tidal marshes, that is not always the case. Since we are interested in the marsh edge dynamics in the absence of mussels, in this paper we focus on the related model without the mussel population. The goal is to understand which conditions lead to a spatially varying marsh shoreline (Figure \ref{fig:1}a) versus a spatially uniform marsh shoreline (Figure \ref{fig:1}b) and what may be the implications of this spatial heterogeneity.

In general, there are three classes of deterministic models that explain pattern formation as a result of scale-dependent feedbacks: Turing-style activator inhibitor models, kernel-based models and differential flow models \citep{borgogno2009mathematical}. In Turing-style activator inhibitor models, the pattern formation arises as a result of differences in the diffusion coefficients of the activator and inhibitor species \citep{turing1952chemical,white1998spatial,parshad2014turing}. Kernel-based models are typically integro-differential equations where the pattern formation arises from the spatial interactions modeled using a kernel function, describing the nature of the short-range and long-range interactions \citep{ britton1990spatial,gourley2001spatio,murray2001mathematical, billingham2003dynamics,ninomiya2017reaction}. This kernel-based approach is a common feature in neural models \citep{amari1977dynamics}, and has also been used in models of vegetation patterns in arid and semi-arid climates \citep{d2006patterns,borgogno2009mathematical,merchant2011instabilities,martinez2013vegetation,martinez2014minimal,martinez2018scale}. Finally, differential flow models are similar to Turing models but the pattern formation now arises not just from the differences in diffusion coefficients, but also from the differences in the flow rates of the species, reflected in the additional advection terms \citep{rovinsky1993self,siero2015striped,klausmeier1999regular}. 

The model we propose here combines elements of both the Turing-model and the kernel-based model and includes both diffusion terms and a kernel function that describes the short-range and long-range interactions between marsh grass and sediment volume. On a local scale, marsh grass attenuates hydrodynamic energy, enhancing sediment accretion and promoting further vegetation growth while the diverted water flow promotes formation of erosion troughs over longer distances \citep{bouma2007spatial,balke2012conditional, schwarz2015interactions}. We model this scale-dependent feedback using a Mexican-hat kernel function that quantifies the strength of positive and negative feedbacks neighboring individuals exert on each other \citep{fuentes2003nonlocal,d2006patterns,borgogno2009mathematical,Siebert15}.
Similar kernel-based approaches have been used to model nonlocal interactions in the context of predator-prey and competition dynamics \citep{merchant2011instabilities,bayliss2015patterns, banerjee2016prey}. The interactions in our system are mostly cooperative and the impact of nonlocal interactions in such systems have not been studied in depth. Given the importance of facilitation in ecosystem dynamics \citep{bertness1994positive, halpern2007incorporating,silliman2015facilitation,he2013global}, it therefore becomes imperative to study nonlocal interactions in cooperative systems. In addition, cooperative systems are likely to display bistable dynamics and the phenomenon of hysteresis \citep{kefi2016can,van2001alternate}. This makes such systems especially prone to collapsing to an irreversible state as environmental conditions gradually worsen and a tipping point is reached \citep{dakos2011slowing, kefi2014early, kefi2016can}. Pattern formation has previously been suggested as a possible coping mechanism, allowing such systems to escape degradation past their tipping point \citep{chen2015patterned}. Due to the reported degradation of tidal marsh habitats around the world, the study of pattern formation in these systems becomes particularly important and can provide more insight into the possible pattern forming mechanism and its implication for the system's resilience and adaptation to environmental changes.

Our paper is organized as follows: In Section 2, we introduce the nonlocal reaction-diffusion model and the background from ecological literature. Section 3 includes analysis and simulation results. By approximating our model using a steady state biharmonic approximation, we are able to derive conditions for the emergence of spatial patterns in our system. We then use numerical simulations to confirm our theoretical findings. Some concluding remarks are made in Section 4.

 \begin{figure}
 	\centering
 	\includegraphics[trim=4.8cm 9cm 5cm 9.6cm,width=0.55\textwidth,clip]{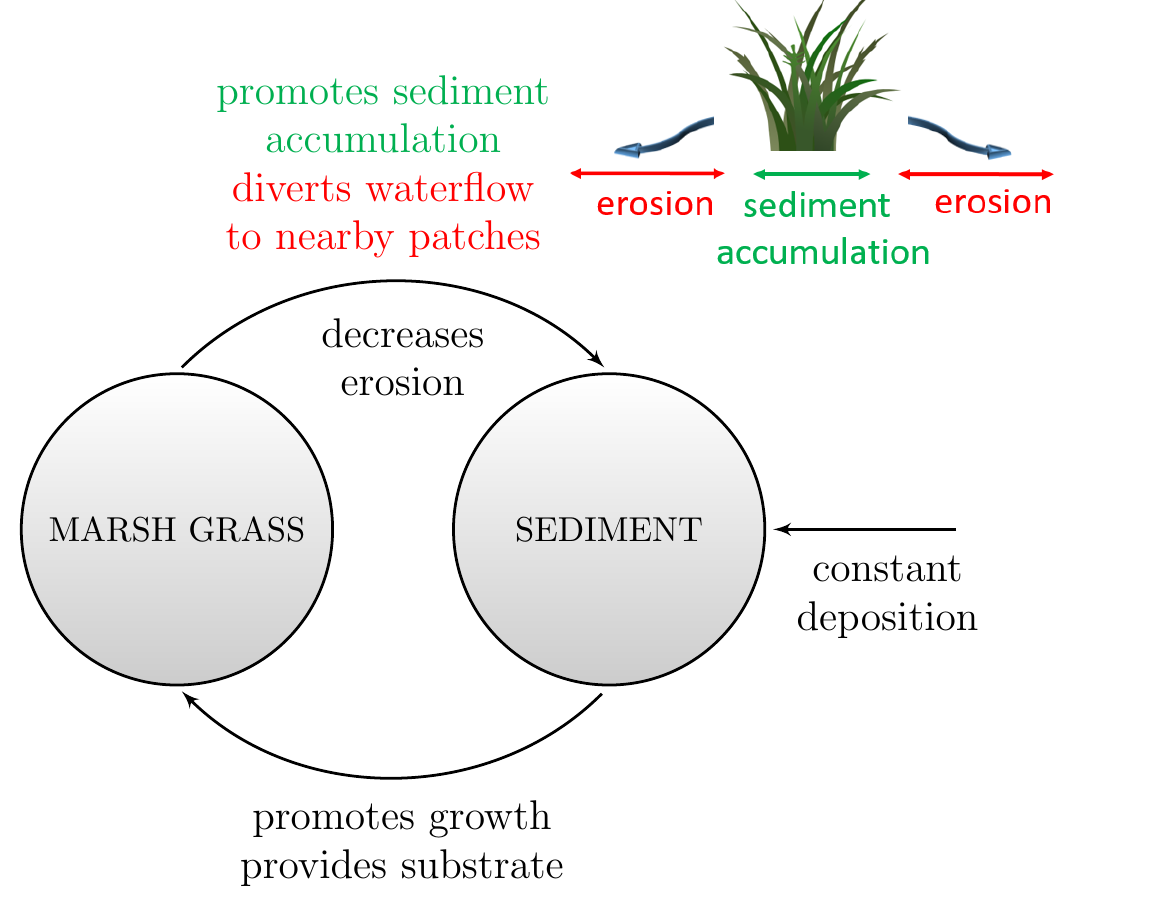}
 	\caption{Diagram of grass-sediment interactions adapted from \citep{Bertness84}.}  \label{fig:2}
 \end{figure}

\section{Model}
 \label{sec:1}
 We consider the two-way interactions between marsh grass and sediment (Figure \ref{fig:2}). Marsh grass binds sediment, stabilizes the marsh edge and attenuates wave energy, helping to mitigate effects of erosion \citep{Gleason79,gedan2011present,ysebaert2011wave,moller2014wave}. As a consequence of reduced erosion, the increased sediment levels promote vegetation growth by decreasing tidal currents \citep{nyman1993relationship, van2005self}. Along with these local interactions, there is a nonlocal interaction that occurs between marsh vegetation and sediment \citep{k2008does, schwarz2015interactions,bouma2009density,VanHulzen07}. Over short distances, marsh vegetation enhances sediment accretion through the attenuation of hydrodynamic energy, contributing to short-range activation. However, as the water gets diverted to the surrounding areas, those areas erode more quickly, contributing to long-range inhibition \citep{ bouma2007spatial,balke2012conditional, bouma2013organism,  fagherazzi2013marsh, fagherazzi2014coastal}. 
 
 Incorporating all the above mentioned interactions, we obtain the following system:
\begin{equation}
\left \{ \begin{aligned}
\partial_{\tau} \hat{G}&=\hat{D}_{\hat{G}}\partial^2_{x}\hat{G}+ \underbrace{\hat{G}\Big(\hat{F}(\hat{S})-c\hat{G}\Big)}_\text{Logistic growth},&x\in \mathbb{R}, \tau>0, \\ 
\partial_{\tau}\hat{S}&=\hat{D}_{\hat{S}}\partial^2_{x}\hat{S}+\underbrace{\eta}_\text{Deposition}-\underbrace{\hat{S}\hat{L}(\hat{G})}_\text{Erosion}+\underbrace{\hat{\lambda} \hat{S}\int_{-\infty}^{\infty} P(x') \hat{G}(x-x') dx'}_\text{Nonlocal deposition/erosion}, &x\in \mathbb{R}, \tau>0,\\
\hat{G}(x,0)&=\hat{G_{0}}(x,0)\ge 0,\,\,\, \hat{S}(x,0)=\hat{S_{0}}(x,0)\ge 0, &x\in \mathbb{R}.
\end{aligned}\right.
\label{eq:1}
\end{equation}
where
\begin{equation*}
\begin{aligned}
\hat{F}(\hat{S})=\frac{p^*(\hat{S}-l_{1})}{\hat{S}+l_{1}^*},\;\;\hat{L}(\hat{G})=\frac{\psi(\hat{G}+k_{s}g)}{\hat{G}+k_{s}},
\end{aligned}
\end{equation*} 
with
\begin{equation*}
\begin{aligned}
p^*,c,l_{1},l_{1}^*, \psi,k_{s},g,\eta, \hat{\lambda} \ge 0.
\end{aligned}
\end{equation*} 
We consider the change in grass shoot density $\hat{G}(x,t)$ ($shoots/m^2$) and sediment height $\hat{S}(x,t)$ ($meters$) on an infinite domain with $x \in \mathbb{R}$, which represents the one-dimensional horizontal cross-section of the marsh edge (see Figure \ref{fig:3}a). We assume logistic growth for the grass density and make an adjustment for the obligatory nature of grass-sediment interactions where below some minimum sediment height $l_{1}$, grass cannot persist. For the sediment equation, we include the baseline sediment deposition $\eta$ \citep{van2005scale,liu2012alternative,Liu14a}. The erosion term is a decreasing function of grass density with $g > 1$ where $\psi g$ corresponds to the minimum erosion rate in the total absence of grass \citep{mariotti2010numerical,silliman2012degradation}. In addition, each equation also includes a diffusion term to quantify spread along the shoreline with diffusion coefficients $\hat{D}_{\hat{G}}$ and $\hat{D}_{\hat{S}}$. To model the scale-dependent interactions, we use a convolution term with a Mexican-hat kernel function $P(x)$:
\begin{equation}
\begin{aligned}
P(x)=\frac{1}{\sqrt{2 \pi}}\Big[\frac{1}{\sigma_{1}}\exp\Big(-\frac{x^2}{2\sigma_{1}^2}\Big)-\frac{1}{\sigma_{2}}\exp\Big(-\frac{x^2}{2\sigma_{2}^2}\Big)\Big],\,\,
\sigma_{1}<\sigma_{2}.
\label{kernel}
\end{aligned}
\end{equation}
The choice of the kernel function is appropriate given the nature of the scale-dependent feedback with short-range positive interactions and long-range negative interactions  (Figure \ref{fig:3}b). There are three main parameters that control the shape of the kernel: $\hat{\lambda}$, which modulates the amplitude and variances $\sigma_{1}^2$ and $\sigma_{2}^2$, which specify the scale of the excitatory and inhibitory interactions, respectively. Further, the kernel function $P(x)$ is symmetric and satisfies the following property: 
\begin{equation}
\begin{aligned}
\int_{-\infty}^{\infty} P(x)dx=0.
\end{aligned}
\label{eq:symmetry}
\end{equation} 

\begin{figure}
\centering
\includegraphics[width=0.75\textwidth]{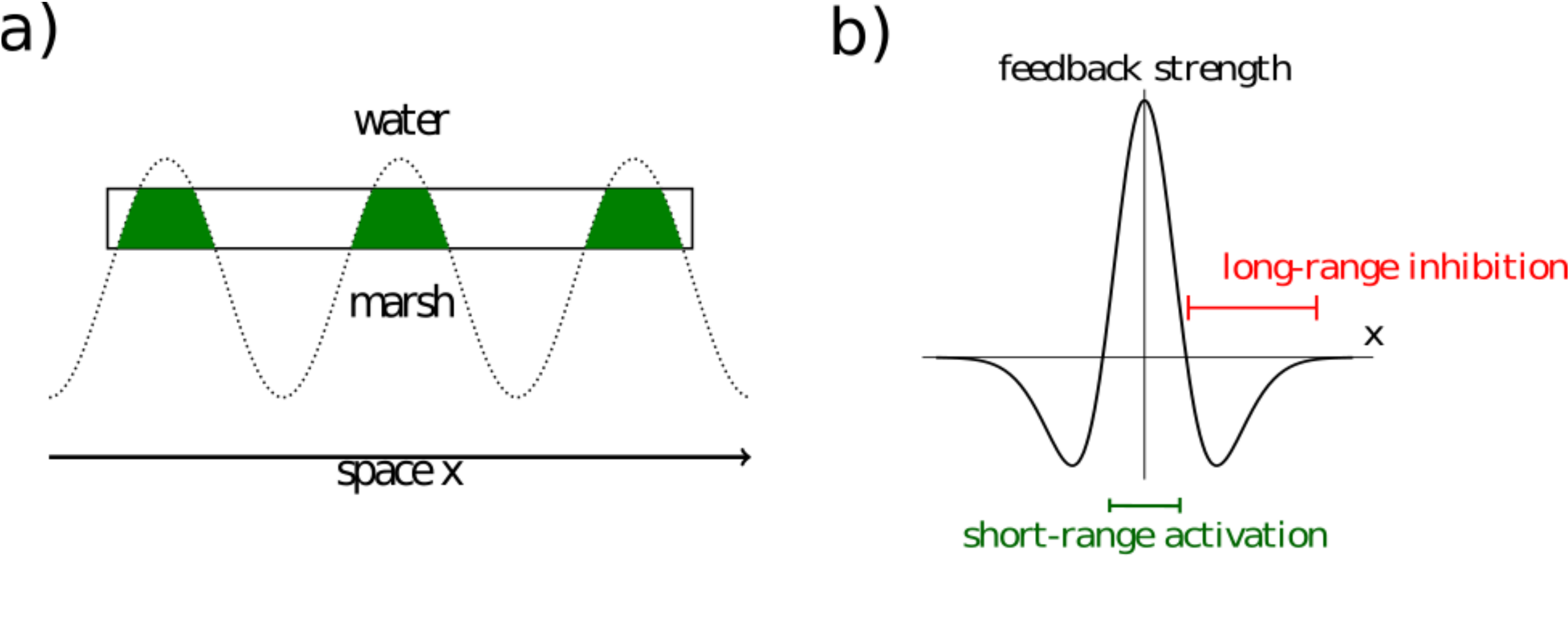}
\caption{Illustration of a) the cross-section of marsh edge used to model the marsh dynamics and b) Mexican-hat kernel and scale-dependent feedback adapted from \citep{Rietkerk08}.}
\label{fig:3}
 \end{figure}
For mathematical simplification, we non-dimensionalize system (\ref{eq:1}) by using the following rescaling:
\begin{equation*}
\begin{aligned}
t&=p^{*}\tau,\ G=\frac{c}{p^{*}}\hat{G},\, 
 \, S=\frac{\psi g}{\eta}\hat{S}.  
\end{aligned}
\end{equation*} 
Then the original system (\ref{eq:1}) becomes:
\begin{equation}
\left \{ \begin{aligned}
\partial_{t}G&=D_{G}\partial^2_{x}G+G\Big(F(S)-G\Big),&x\in  \mathbb{R}, t>0,  \\
\partial_{t}S&=D_{S}\partial^2_{x}S+\phi\Big(-L(G) S+1\Big)+\lambda S\int_{-\infty}^{\infty} P(x') G(x-x') dx',&x\in  \mathbb{R}, t>0,\\
G(x,0)&=G_{0}(x,0)\ge 0,\,\,\, S(x,0)=S_{0}(x,0)\ge 0, &x\in  \mathbb{R}.
\end{aligned}\right.
\label{eq:2}
\end{equation}
with
\begin{equation}
\begin{aligned}
F(S)&=\frac{S-e_{1}}{S+p_{1}},\,\, L(G)=\frac{\delta G+e_{3}}{G+e_{3}},
\end{aligned}
\label{eq:scaled_functions}
\end{equation}
and $P(x)$ still defined as before. The new parameters are all positive with the following rescaling:
\[\begin{matrix}
\ds e_{1}=\frac{\psi g l_{1}}{\eta},& \ds p_{1}=\frac{\psi g l_{1}^*}{\eta}, &  \ds e_{3}=\frac{k_{s}c}{p^{*}}, & \ds \delta=\frac{1}{g} \\
\ds \phi=\frac{\psi g}{p^{*}}, & \ds D_{G}=\frac{\hat{D}_{G}}{p^{*}}, &  \ds D_{S}=\frac{\hat{D}_{S}}{p^{*}}, 
& \ds \lambda=\frac{\hat{\lambda}}{c}. &
\end{matrix}\]
Not only does this rescaling simplify the notation, but it also allows for an easier interpretation of the functional forms of $F(S)$ and $L(G)$ (See Figure \ref{fig:10} in the Appendix). The scaled intrinsic growth rate of grass is now between $0$ and $1$, and we can think of the threshold $e_{1}$ as the minimum amount of sediment necessary for the persistence of grass. Similarly, the erosion term given by $L(G)$ is scaled to be between $\delta$ and $1$ for ease of interpretation. 

\section{Results} \label{sec:3}

In classic Turing models, spatially patterned solutions result from symmetry-breaking instability in which an otherwise stable spatially uniform steady state can become destabilized by the addition of diffusion and lead to the emergence of spatial patterns. The condition for the emergence of spatial patterns is contingent on the idea that the species in the model diffuse at significantly different rates, with the activator species diffusing much more slowly than the inhibitor species. The conditions for such a Turing-instability can be derived by performing a linear stability analysis around the positive steady state to obtain conditions under which the addition of diffusion acts to destabilize the system. Our model differs from the classic Turing model in that it lacks the classic activator-inhibitor dynamics and includes an additional kernel function term that models the scale-dependent feedback between grass and sediment volume. Assuming that the kernel function has a limited effect at relatively large distances, we can perform a biharmonic approximation of our system and decompose the integral term into two terms involving just partial derivatives, corresponding to short-range positive interactions and long-range negative interactions. We can then perform a linear stability analysis around the positive steady state and derive conditions under which this state is destabilized and leads to the emergence of a spatially periodic solution. Therefore, we first consider the spatially independent dynamics of our model and derive conditions under which the positive steady state is stable in the corresponding system of ODEs and then use these results to understand the spatial dynamics of the full model.

\subsection{Spatially homogeneous model}\label{sec:3.1}

Let's assume that $G$ and $S$ do not vary and are spatially constant. Then, we can use the property in (\ref{eq:symmetry}) and drop both the diffusion and integral terms. In this way, we are left with the following spatially independent system:
\begin{equation}
\left \{ \begin{aligned}
\frac{dG}{dt}&=G\Big(F(S)-G\Big), & t>0, \\
\frac{dS}{dt}&=\phi(-L(G) S+1), &t>0.
\end{aligned} \right.
\label{eq:3}
\end{equation}
We look for spatially uniform steady states $(G^{*}, S^{*})$ of (\ref{eq:3}) which satisfy $\frac{dG}{dt}=0$ and $\frac{dS}{dt}=0$. There are two such types of steady states: the degraded (grass-free) state $E_{S}=(0,1)$ and the coexistence state $E_{GS}=(G^*,\frac{1}{L(G^*)})$ with both grass and sediment present, where $G^*$ satisfies $G=F(\frac{1}{L(G)})$. Since we are interested in physically realistic positive steady states, the coexistence state $E_{GS}$ exists if and only if  $\frac{1}{L(G^*)}>e_{1}$.

We first consider the degraded state $E_{S}=(0,1)$ and its stability. This result is summarized below.
\begin{proposition}\label{proposition}
	The degraded steady state $E_{S}=(0,1)$ is locally asymptotically stable with respect to (\ref{eq:3}) if $ e_{1}>1$ and is unstable with respect to (\ref{eq:3}) if $ e_{1}<1$.
\end{proposition}
\begin{proof}
	 The Jacobian matrix $\mathbf{J}$ of (\ref{eq:3}) evaluated at $E_{S}=(0,1)$ is given by:
	\[
	\mathbf{J}_{E_{S}}=
	\begin{bmatrix}
	F(1)  & 0   \\
	-\phi \frac{dL}{dG} & -\phi L(0)    
	\end{bmatrix}.
	\]
	The two corresponding eigenvalues are $\lambda_{1}=F(1)=\frac{1-e_{1}}{1+p_{1}}$ and $\lambda_{2}=-\phi L(0)=-\phi$. It is clear that $\lambda_{2}$ is always negative. Further, $\lambda_{1}=\frac{1-e_{1}}{p_{1}+1}$ is negative for $e_{1}>1$. Therefore, the steady state $E_{S}$ is locally asymptotically stable for $e_{1}>1$ and unstable for $e_{1}<1$.  
\end{proof}
The parameter $e_{1}$ is the minimal steady state sediment elevation needed for the persistence of grass. For the trivial steady state $E_{S}=(0,1)$, as long as $e_{1}>1$, its value will exceed the steady state value of sediment, leading to a negative growth rate for grass and a stable trivial state. 

We now consider the positive coexistence steady state $E_{GS}=(G^*,\frac{1}{L(G^*)})$ where $G^*$ satisfies $G=F(\frac{1}{L(G)})$ and obtain the following theorem:
\begin{theorem}\label{theorem1}
	Suppose that $p_{1},e_{1}, e_{3}, \phi>0$ and $0<\delta<1$. Let
	\begin{equation}
	\begin{aligned}
	A=1+p_{1}\delta,\,	B=-1+e_{3}+p_{1}e_{3}.
	\end{aligned}
	\label{eq:AandB}
	\end{equation}
	\begin{enumerate}
		\item \textbf{(Case I)}\, If $B+\delta<0$, then there exists a saddle-node bifurcation point $e_{1}=e_{1}^*>1$ such that (\ref{eq:3}) has one positive steady state $(G^*_{+},\frac{1}{L(G^*_{+})})$ for $0<e_{1}\le 1$ and $e_{1}=e_{1}^*$, two positive steady states $(G^*_{\pm},\frac{1}{L(G_{\pm})})$ for $1<e_{1}<e_{1}^*$,   and no positive steady state for $e_{1}>e_{1}^*$. The bifurcation point $e_{1}^*$ is defined as follows:
		\begin{equation}
			\begin{aligned}
		e_{1}^*=\frac{2e_{3}A(\delta^2+B\delta-Ae_{3})+(2Ae_{3}-B\delta)\sqrt{A^2e_{3}^2-Ae_{3}\delta(B+\delta)}}{\delta^2\sqrt{A^2e_{3}^2-Ae_{3}\delta(B+\delta)}}.
		\end{aligned}
		\label{eq:bifurcation_point}
		\end{equation}
		\item \textbf{(Case II)}\,
		If $B+\delta \ge 0$, then there exists a unique positive steady state $(G^*_{+},\frac{1}{L(G^*_{+})})$ for all $0<e_{1}<1$, and no positive steady state for $e_{1}\ge 1$.
	\end{enumerate}
\end{theorem}
\begin{proof}
We can rewrite $G=F(\frac{1}{L(G)})$ as
\begin{equation}
\begin{aligned}
G&=\frac{\frac{1}{L(G)}-e_{1}}{\frac{1}{L(G)}+p_{1}}\\
&\implies e_{1}=\frac{1-G}{L(G)}-p_{1}G\\
&\implies e_{1}=\frac{e_{3}-BG-AG^2}{G\delta+e_{3}}:=K(G),
\end{aligned}
\label{eq:parabola2}
\end{equation} 
where $A$ and $B$ are defined as in (\ref{eq:AandB}).
The function $K(G)$ from (\ref{eq:parabola2}) crosses the horizontal axis at
\begin{equation}
\begin{aligned}
G^{\pm}_{K}=\frac{-B\pm \sqrt{B^2+4e_{3}A}}{2A}.
\end{aligned}
\label{eq:roots}
\end{equation}
Since $A\ge 0$, the roots in (\ref{eq:roots}) have to be of opposite sign. Therefore, the graph of $K(G)$ has one positive and one negative root. Note that the vertical asymptote of $K(G)$ is irrelevant as it is located where $G$ is negative and outside of the physically realistic range. 

Differentiating $K(G)$ in (\ref{eq:parabola2}) with respect to $G$ yields:
\begin{equation}
\begin{aligned}
K'(G)&=\frac{-B-2AG-\delta K(G)}{\delta G+e_{3}}=\frac{-AG(2e_{3}+\delta G)-e_{3}(B+\delta)}{(\delta G+e_{3})^2}\\
&=\frac{-L(G)-L'(G)(1-G)}{L(G)^2}-p_{1},\\
\label{eq:slope}
\end{aligned}
\end{equation}
and 
\begin{equation}
\begin{aligned}
K(0)=1,\, K'(0)=\frac{-(B+\delta)}{e_{3}},
\end{aligned}
\label{eq:slope2}
\end{equation}
Further, we can set $K'(G)=0$ to obtain the maximum and minimum points of the function:
\begin{equation}
\begin{aligned}
\tilde{G}_{\pm}=\frac{-e_{3} \pm \sqrt{e_{3}^2-\frac{e_{3}\delta(B+\delta)}{A}}}{\delta}.
\end{aligned}
\label{eq:turning}
\end{equation}
We then have two cases arising depending on the sign of $K'(0)$ in (\ref{eq:slope2}) (Figure \ref{fig:4}). 
\begin{figure}
	\centering
	\includegraphics[trim=5.2cm 19.5cm 7.2cm 4.4cm, width=0.70\textwidth,clip]{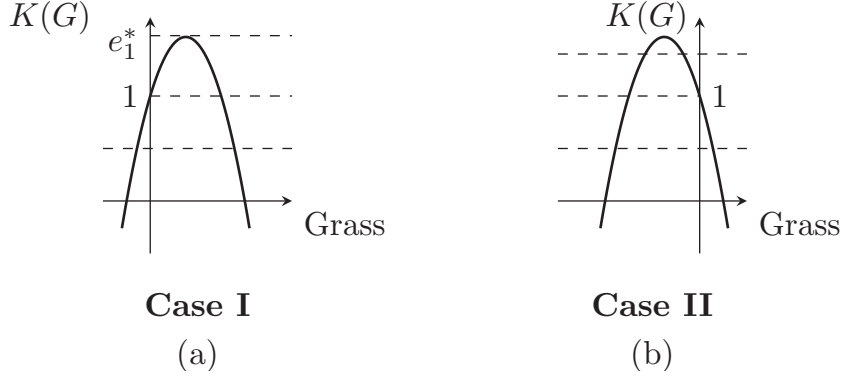}
	\caption{Schematics representations of parameter regimes for the positive coexistence steady state $E_{GS}$. Case I corresponds to the scenario where marsh vegetation is very efficient at reducing erosion and $e_{3}<\frac{1-\delta}{1+p_{1}}$. The horizontal values represent various values of $e_{1}$. We see that for $e_{1}>1$, we have two real, positive steady states. They eventually collide and disappear in a saddle node bifurcation $e_{1}^{*}$ . For Case II, the case of less efficient vegetation, we see that for $e_{1}\ge 1$, there are no positive steady states and for $e_{1}<1$, there is only one. }
	\label{fig:4}
\end{figure}
\\ \textbf{Case I}.
The first case corresponds to $B+\delta<0$ and a more physically realistic parameter regime where grass is more effective at attenuating erosion. In this parameter regime, $\delta$ is smaller and therefore, the erosion rate decays faster as a function of grass. From (\ref{eq:turning}), it is clear that $\tilde{G}_{-}<0<\tilde{G}_{+}$ and there exists only one peak for positive values of G, given by the value of $\tilde{G}_{+}$. Further, since $K(0)=1$, this means that for $e_{1}<1$, there exists only one positive steady state and for $e_{1}>1$, there exist two positive steady states  (Figure \ref{fig:4}a). The two positive steady states collide and annihilate each other at the saddle-node bifurcation point $e_{1}^*$, given by:
\begin{equation*}
\begin{aligned}
e_{1}^*=K(\tilde{G}_{+})= \frac{2e_{3}A(\delta^2+B\delta-Ae_{3})+(2Ae_{3}-B\delta)\sqrt{A^2e_{3}^2-Ae_{3}\delta(B+\delta)}}{\delta^2\sqrt{A^2e_{3}^2-Ae_{3}\delta(B+\delta)}}.
\end{aligned}
\end{equation*}
\\ \textbf{Case II} Case II corresponds to $B+\delta\ge 0$, a parameter regime in which cordgrass is less effective at attenuating sediment erosion. From (\ref{eq:turning}), it is clear that $\tilde{G}_{-}<\tilde{G}_{+}<0$ and there exist no peaks for positive values of G. Therefore, since $K(0)=1$, for $e_{1}<1$, we have one positive steady state, while for $e_{1}>1$ there is no positive steady state (Figure \ref{fig:4}b).
\end{proof}
Now that we know how many positive steady states can be expected, we evaluate their stability and obtain the following theorem:
\begin{theorem}\label{theorem2}
	Suppose that $p_{1},e_{1}, e_{3}, \phi>0$, $0<\delta<1$ and let $A$ and $B$ be defined as in (\ref{eq:AandB}). For the positive steady states $E_{GS}=(G^*_{\pm},S^*_{\pm})$ defined as: 
\begin{equation}
\begin{aligned}
G_{\pm}^*&=\frac{-(e_{1}\delta+B)\pm \sqrt{(e_{1}\delta+B)^2-4Ae_{3}(e_{1}-1)}}{2A},\\
S_{\pm}^*&=\frac{G_{\pm}^*+e_{3}}{\delta G_{\pm}^*+e_{3}},
\end{aligned}
\label{eq:steady}
\end{equation}
we have the following cases:
	\begin{enumerate}
		\item \textbf{(Case I)}\, Let $B+\delta<0$. For $1<e_{1}<e_{1}^*$ where $e_{1}^*$ is defined as in (\ref{eq:bifurcation_point}),	the high density positive steady state $(G^{*+},\frac{1}{L(G_{+}^*)})$ is locally asymptotically stable and the low density positive steady state $(G_{-}^*,\frac{1}{L(G_{-}^*)})$ is unstable. For $0<e_{1}\le 1$, there is only one positive steady state $(G^{*+},\frac{1}{L(G_{+}^*)})$ which is locally asymptotically stable.
		\item \textbf{(Case II)}\,  Let  $B+\delta \ge 0$. Then, for all $0<e_{1}<1$, the unique positive steady state $(G_{+}^*,\frac{1}{L(G_{+}^*)})$ is locally asymptotically stable.
	\end{enumerate}
\end{theorem}
\begin{proof}
We first evaluate the Jacobian matrix $\mathbf{J}$ of system (\ref{eq:3}) at the positive steady state $E_{GS}$. This is given by:
\[
\mathbf{J}(E_{GS})=
\begin{bmatrix}
-G^* & G^* F'(\frac{1}{L(G^*)})   \\
-\phi \frac{1}{L(G^*)}L'(G^*)& -\phi L(G^*)
\end{bmatrix}.
\]
This is just the general form of the Jacobian evaluated at the positive steady state type. From Theorem \ref{theorem1}, we can have either two such positive states (high and low) or just one, depending on the parameter regime. We will consider both cases in this proof. We note the special form of the Jacobian matrix, reflecting the cooperative nature of our system:
\[
\mathbf{J}=
\begin{bmatrix}
- & +   \\
+& -
\end{bmatrix}.
\]
From $\mathbf{J}$, we can define the trace and determinant as follows:
\begin{equation*}
\begin{aligned}
\Tr{\mathbf{J}}&=-G^* -\phi L(G^*),\\
\Det{\mathbf{J}}&=\phi G^* L(G^*)+\phi \frac{1}{L(G^*)} G^* L'(G^*)F'\Big(\frac{1}{L(G^*)}\Big).
\end{aligned}
\end{equation*}
In order for $E_{GS}$ to be locally asymptotically stable, we need $\Tr{\mathbf{J}}<0$ and $\Det{\mathbf{J}}>0$. Since $L(G^*) \ge 0$ and $G^*$ is a positive quantity, the trace of $\mathbf{J}$ is always negative. Note that since $\Tr{\mathbf{J}} <0$, a Hopf bifurcation cannot occur from the positive steady state. 
Therefore, to assess stability, we need to determine the sign of $\Det{\mathbf{J}}$. Using (\ref{eq:parabola2}) and (\ref{eq:slope}), we can rewrite $\Det{\mathbf{J}}$ in terms of $K'(G)$ to obtain: 
\begin{equation*}
\begin{aligned}
\Det{\mathbf{J}}&=\phi G^*\Big(L(G^*)+\frac{1}{L(G^*)} L'(G^*)F'(\frac{1}{L(G^*)})\Big)\\
&=\phi G^*\Big(L(G^*)+\frac{1}{L(G^*)} L'(G^*)\frac{p_{1}+e_{1}}{(\frac{1}{L(G^*)} +p_{1})^2}\Big)\\
&=\phi G^*\Big(L(G^*)+\frac{1}{L(G^*)} L'(G^*)\frac{p_{1}+\frac{1-G^*}{L(G^*)} -p_{1}G^*}{(\frac{1}{L(G^*)} +p_{1})^2}\Big)\\
&=\phi G^*\Big(L(G^*)+\frac{L'(G^*)(1-G^*)}{1+L(G^*)p_{1}}\Big)\\
&=\phi G^*\Big(\frac{-K'(G^*)L^2(G^*)}{1+L(G^*)p_{1}}\Big).
\end{aligned}
\label{eq:determinant}
\end{equation*}
From equation (\ref{eq:parabola2}), we can solve the steady states explicitly as in \eqref{eq:steady}.
We now consider two cases from Theorem \ref{theorem1}. For Case I, both low and high steady states $(G_{+}^*,S_{+}^*)$ and  $(G_{-}^*,S_{-}^*)$ are positive, while for Case II, only the high positive steady state $(G_{+}^*,S_{+}^*)$ is positive. These are the steady states we consider and assess their stability.
\\ \textbf{Case I} For Case I ($B+\delta < 0$), we have the following scenarios:
\begin{itemize}
    \item \textbf{(i)} For $1<e_{1}< e_{1}^*$, there are two positive steady states  $(G_{+}^*,S_{+}^*)$ and $(G_{-}^*,S_{-}^*)$. From the definitions of the steady states in (\ref{eq:steady}), it follows that
   \begin{equation} 
   \begin{aligned}
   -e_{1}\delta-B&>\sqrt{(e_{1}\delta+B)^2-4Ae_{3}(e_{1}-1)}>0,\\
   (e_{1}\delta+B)^2&>4A(e_{1}e_{3}-e_{3}).
   \end{aligned}
   \label{cond_case1}
    \end{equation}
    Evaluating $K'(G)$ from equation (\ref{eq:slope}) at $G=G_{+}^*$ and $G=G_{-}^*$  yields:
\begin{equation}
\begin{aligned}
K'(G_{+}^*)=\frac{-2\delta(-e_{1}\delta-B)^2-2\delta (-e_{1}\delta -B) C +8 \delta A e_{3}(e_{1}-1)-4Ae_{3}C}{4A(\delta G_{+}^*+e_{3})^2},\\
K'(G_{-}^*)=\frac{-2\delta(-e_{1}\delta-B)^2+2\delta (-e_{1}\delta -B) C +8 \delta A e_{3}(e_{1}-1)+4Ae_{3}C}{4A(\delta G_{-}^*+e_{3})^2},
\end{aligned}
\label{eq:Keval}
\end{equation}
with 
    \begin{equation*}
C=\sqrt{(e_{1}\delta+B)^2-4Ae_{3}(e_{1}-1)}>0.
\end{equation*}
Using conditions from (\ref{cond_case1}), we can show:
\begin{equation}
\begin{aligned}
K'(G_{+}^*)&=\frac{-2\delta(-e_{1}\delta-B)^2-2\delta (-e_{1}\delta -B) C +8 \delta A e_{3}(e_{1}-1)-4Ae_{3}C}{4A(\delta G_{+}^*+e_{3})^2}\\
&<\frac{-2\delta(-e_{1}\delta-B)^2-2\delta (-e_{1}\delta -B) C +2 \delta (e_{1}\delta+B)^2-4Ae_{3}C}{4A(\delta G_{+}^*+e_{3})^2}\\
&=\frac{-2\delta (-e_{1}\delta -B) C-4Ae_{3}C}{4A(\delta G_{+}^*+e_{3})^2}<0,
\end{aligned}
\label{eq:Kplus}
\end{equation}
and
\begin{equation}
\begin{aligned}
K'(G_{-}^*)&=\frac{-2\delta(-e_{1}\delta-B)^2+2\delta (-e_{1}\delta -B) C +8 \delta A e_{3}(e_{1}-1)+4Ae_{3}C}{4A(\delta G_{-}^*+e_{3})^2}\\
&>\frac{-2\delta(-e_{1}\delta-B)^2+2\delta C^2 +8 \delta A e_{3}(e_{1}-1)+4Ae_{3}C}{4A(\delta G_{-}^*+e_{3})^2}\\
&=\frac{4Ae_{3}C}{4A(\delta G_{-}^*+e_{3})^2}>0.
\end{aligned}
\label{eq:Kminus}
\end{equation}
Therefore, since $K'(G)<0$ on the $(G_{+}^*,S_{+}^*)$ branch, $\Det{\mathbf{J}}$ evaluated at $(G_{+}^*,S_{+}^*)$ is positive and $(G_{+}^*,S_{+}^*)$ is locally asymptotically stable.
Similarly, since $K'(G)>0$ on the $(G_{-}^*,S_{-}^*)$ branch, $\Det{\mathbf{J}}$ evaluated at $(G_{-}^*,S_{-}^*)$ is negative and $(G_{-}^*,S_{-}^*)$ is unstable.
 \item \textbf{(ii)} For $0<e_{1}\le 1$, there is only one positive steady state branch corresponding to $(G_{+}^*, \frac{1}{L(G_{+}^*)})$. Further, we can show that
 \begin{equation*}
0<-\delta-B<-e_{1}\delta-B.
\end{equation*}
From (\ref{eq:Keval}), it then follows that $K'(G_{+}^*)<0$. Since $\Det{\mathbf{J}}$ evaluated at $(G_{+}^*,S_{+}^*)$ is positive, $(G_{+}^*,S_{+}^*)$ is locally asymptotically stable.
\end{itemize} 
\textbf{Case II} For Case II ($B+\delta \ge 0$) , there is a unique positive steady state branch corresponding to $(G_{+}^*, S_{+}^*)$. From equation (\ref{eq:slope}) it is clear that $K'(G)<0$ for all positive values of $G$. Therefore, since $\Det{\mathbf{J}}$ evaluated at $(G_{+}^*,S_{+}^*)$ is positive, the steady state $(G_{+}^*,S_{+}^*)$ is locally asymptotically stable.
\end{proof}
\begin{figure}
	\centering
     \includegraphics[width=1\textwidth]{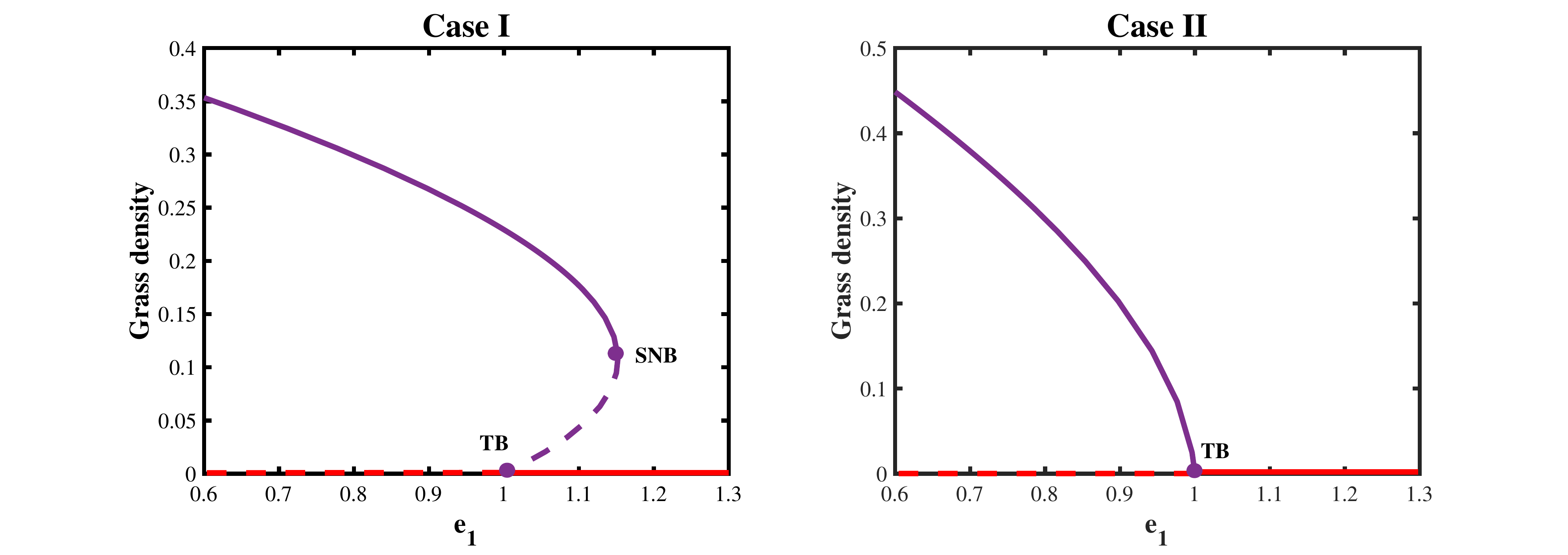}
	\caption{Bifurcation diagrams plotted using MatCont \citep{matcont} for Case I with a saddle-node (SNB) bifurcation happening at $e_{1}=1.16$ and a transcritical (TB) bifurcation happening at $e_{1}=1$ and Case II with only a transcritical bifurcation happening at $e_{1}=1$. Parameters used in Case I: $p_{1}=3.5,\phi=0.14,e_{3}=0.1140,\delta=1/7$. Parameters used in Case II: $p_{1}=0.5,\phi=0.14,e_{3}=0.5,\delta=0.3$.}
	\label{fig:5}
\end{figure}
The results from Proposition \ref{proposition}, Theorem \ref{theorem1} and Theorem \ref{theorem2} are summarized in Figure \ref{fig:5}. In Case I, the system displays bistability for values $1<e_{1}<e_{1}^*$, where both the high positive steady state and the trivial steady state are stable, separated by an unstable positive steady state branch. The two positive steady states then merge in a saddle-node bifurcation at $e_{1}^*$, after which only the stable trivial steady state $E_{S}$ remains. Bistability is not surprising given the highly cooperative nature of this system and large role that the grass plays in erosion mitigation. In Case II, which corresponds to the scenario where grass is less effective at attenuating erosion, the unique stable positive state gradually decreases and eventually undergoes a transcritical bifurcation at $e_{1}=1$ at which it exchanges stability with the trivial steady state $E_{S}$. Note that the positive steady state in Case II ceases to exist for smaller values of $e1$ than in Case I. This is intuitive as Case I corresponds to a more cooperative parameter regime that makes population persistence more possible.

\subsection{Generalized Cooperative System with Nonlocal Interactions}
\label{sec:5}
We now consider the spatially extended system to investigate the emergence of a patterned solution. Given the complexity of the spatially extended system (\ref{eq:2}), we carry out a steady state biharmonic approximation of this system, allowing us to perform linear stability analysis on the approximated system and gain insight into the dynamics of the original system (\ref{eq:2}) \citep{murray2001mathematical, d2006patterns, borgogno2009mathematical}.

Let's consider the following general form of our system (\ref{eq:2}):
\begin{equation}
\left \{ \begin{aligned}
\partial_t u&=d_{11}\partial^2_x u+f(u,v), &x\in \mathbb{R}, t>0,\\
\partial_t v&=d_{22}\partial^2_x v+g(u,v)+\lambda v\int_{-\infty}^{\infty} P(x-x') u(x') dx',&x\in \mathbb{R}, t>0,\\
u(x,0)&=u_{0}(x)\ge 0,\,\,\, v(x,0)=v_{0}(x)\ge 0, &x\in \mathbb{R}.
\end{aligned} \right.
\label{eq:subsystem_general}
\end{equation}
where $P(x)$ is defined the same as in (\ref{kernel}), and $f,g$ are general smooth functions.
Following standard procedure, we assume the kernel has a limited effect at relatively large distances and perform a Taylor's expansion of the integral term around $x'=x$ \citep[pages 482-489]{murray2001mathematical}:
\begin{equation*}
\begin{aligned}
&\int_{-\infty}^{\infty} P(x-x') u(x') dx'=\int_{-\infty}^{\infty} P(z) u(x-z) dz\\
=&\int_{-\infty}^{\infty} P(z) \Big[u(x)-z\frac{\partial u(x)}{\partial x}+\frac{z^2}{2!}\frac{\partial^2 u(x)}{\partial x^2}-\frac{z^3}{3!}\frac{\partial^3 u(x)}{\partial x^3}+\frac{z^4}{4!}\frac{\partial^4 u(x)}{\partial x^4}-\cdots \Big]dz.
\end{aligned}
\end{equation*}
This is a reasonable assumption in the context of our model as the scale-dependent grass-sediment feedback is thought to occur on a relatively small spatial scale ($1-4$ meters). We can then define the moments $P_{m}$ in the following way:
\begin{equation*}
\begin{aligned}
P_{m}=\frac{1}{m!}\int_{-\infty}^{\infty}z^{m}P(z)dz, \,\, m=0,1,2,\cdots
\end{aligned}
\end{equation*}
Given the symmetry of the kernel $P(x)$, the odd-power moments vanish, as does $P_{0}$ since $\int_{-\infty}^{\infty} P(x')dx'=0$.
From the specific form of the Mexican-hat kernel in (\ref{kernel}), we can obtain exact expressions for $P_{2}$ and $P_{4}$ in term of the variances $\sigma_{1}$ and $\sigma_{2}$ of the excitatory and inhibitory effects, respectively:
\begin{equation}
\begin{aligned}
P_{2}=\frac{\sigma_{1}^2-\sigma_{2}^2}{2}<0, \,\,P_{4}=\frac{\sigma_{1}^4-\sigma_{2}^4}{8}<0.
\end{aligned}
\label{kernel_definitions}
\end{equation} 
Truncating the expansion at the fourth partial derivative, the original system (\ref{eq:subsystem_general}) can now be approximated by the following biharmonic system \citep{bates1996transition,bates1997heteroclinic,couteron2001periodic}:
\begin{equation}
\left \{ \begin{aligned}
\partial_t u&=d_{11}\partial^2_x u+f(u,v),&x\in \mathbb{R}, t>0,\\
\partial_t v&=d_{22}\partial^2_x v+g(u,v)+\lambda v(P_{2}\partial^2_x u+P_{4}\partial^4_x u),&x\in \mathbb{R}, t>0,\\
u(x,0)&=u_{0}(x,0)\ge 0,\,\,\, v(x,0)=v_{0}(x,0)\ge 0, &x\in \mathbb{R}.
\end{aligned} \right.
\label{eq:subsystem_approx}
\end{equation}
\\In this way, the evolution of $u$ and $v$ now depends not only on their own diffusion as in the classic reaction-diffusion system, but also on the additional short-range cross-diffusion $\partial^2_x u$ and long-range cross-diffusion $\partial^4_x u$ terms. Here, $\lambda P_{2}$ and $ \lambda P_{4}$ represent the corresponding cross-diffusion coefficients. We are interested in the conditions that lead to the emergence of a spatially patterned solution in such a system. In general, spatial patterns arise in such systems through Turing instability, a symmetry breaking mechanism in which an otherwise stable spatially uniform steady state is destabilized by the addition of diffusion and cross-diffusion terms. To derive conditions for such an instability, we perform a classic Turing type linear stability analysis on the approximated system (\ref{eq:subsystem_approx}). 

We expand our system (\ref{eq:subsystem_approx}) about a spatially uniform positive steady state $(u^*,v^*)$ with $u^*>0$ and $v^*>0$. Substituting
\begin{align*}
u(x,t)&=u^{*}+w_{1}(x,t), \,\,\,\,\, |w_{1}(x,t)|\ll u^{*}, \nonumber\\
v(x,t)&=v^{*}+w_{2}(x,t),\,\,\,\,\, |w_{2}(x,t)|\ll v^{*},
\end{align*}
into (\ref{eq:subsystem_approx}) and dropping any nonlinear terms, the resulting linearized system about $(u^*,v^*)$ becomes:
\begin{equation}
\begin{aligned}
\mathbf{W_{t}}=\mathbf{J}\mathbf{W}+\mathbf{D}\nabla^2\mathbf{W}+\mathbf{H}\nabla^4\mathbf{W},
\end{aligned}
\label{eq:matrixsimp}    
\end{equation}
with
\begin{equation}
\mathbf{W}(x,t)=\begin{pmatrix}\, w_{1}(x,t)\,\\\,w_{2}(x,t)\,\end{pmatrix},\,\,\mathbf{D}=\left( \begin{array}{cc}
d_{11}	&  0\\ 
d_{21}	& d_{22} \end{array} \right),\, \mathbf{H}=\left( \begin{array}{cc}
0	&  0\\ 
h_{1}	& 0 \end{array} \right),
\mathbf{J} =	\left( \begin{array}{cc}
f_{u}	&  f_{v}\\ 
g_{u}& g_{v}
\end{array} \right) \Bigg\rvert_{(u^*,v^*)},
\label{eq:assumptions}
\end{equation}
where
\begin{align}
d_{21}=\lambda v^{*}P_{2}<0, \, \, h_{1}=\lambda v^{*}P_{4}<0; \, \, d_{11},d_{22}>0.
\label{eq:diffusionterm}
\end{align}
Here, we consider a cooperative form of $\mathbf{J}$ with $f_{u},g_{v}<0$ and $f_{v},g_{u}>0$:
\begin{equation}
\mathbf{J}=\left( \begin{array}{cc}
-	&  +\\ 
+& -
\end{array} \right),
\label{eq:cooperative_jacobian}
\end{equation}
 Note that this is different from the classic Turing model activator-inhibitor form of  $\mathbf{J}$ where $f_{u}$ and $g_{v}$ are of opposite sign.

Following standard convention, we let
\begin{equation}
\boldsymbol{W}(x,t)= \begin{pmatrix}\, w_{1}(x,t)\,\\\,w_{2}(x,t)\,\end{pmatrix}=\begin{pmatrix}\, a\,\\\,b\,\end{pmatrix}e^{\alpha t+ikx}.
\label{eq:stdconv}     
\end{equation}
Here, $a$ and $b$ are constants, and $k$ is the corresponding wavenumber, with $1/k$ being proportional to the wavelength of the emergent patterns. Since $e^{ikx}$ is periodic and bounded, the sign of $\alpha$ plays an important role in determining whether these small perturbations away from the steady state will grow or decay.
\par Substituting (\ref{eq:stdconv}) into (\ref{eq:matrixsimp}) and
 looking for a nontrivial solution, we require
\begin{equation}
|\alpha \mathbf{I}-\mathbf{J}+k^2\mathbf{D}-k^4\mathbf{H}|=\left| \begin{array}{cc}
\alpha+d_{11}k^2-f_{u}	&  -f_{v}\\ 
-g_{u}+k^2d_{21}-k^4h_{1}	& \alpha+d_{22}k^2-g_{v}
\end{array} \right|=0.
\end{equation}
This yields the following dispersion relation:
\begin{equation}
\alpha^2-b(k^2)\alpha+c(k^2)=0,
\label{eq:dispersion}
\end{equation} 
where
\begin{equation}
\begin{aligned}
b(k^2)&=\Tr{\mathbf{J}}-k^2\Tr{\mathbf{D}},\\
c(k^2)&=(\Det{\mathbf{D}}-f_{v}h_{1})k^4-(d_{11}g_{v}+d_{22}f_{u}-f_{v}d_{21})k^2+\Det{\mathbf{J}}.
\end{aligned}
\label{eq:determ}
\end{equation}
Using this dispersion relation, we can then derive conditions for Turing type instability, summarized in the following theorem:

\begin{theorem}\label{theorem3}
	Let $(u^*,v^*)$ be a constant steady state solution of (\ref{eq:subsystem_approx}) with $\mathbf{D}$ defined as in (\ref{eq:assumptions}) with $d_{11},d_{22}>0$ and $P_{2}<0$ and $P_{4}<0$ defined in (\ref{kernel_definitions}). Also, let $\mathbf{J}$ be defined as in (\ref{eq:cooperative_jacobian}) with $f_{u},g_{v}<0$ and $f_{v},g_{u}>0$. If
	\begin{equation}
	\begin{aligned}
	\Det{\mathbf{J}}&>0,\\
	d_{11}g_{v}+d_{22}f_{u}-f_{v}\lambda v^*P_{2}&>2\sqrt{\Det(\mathbf{D})-f_{v}\lambda v^*P_{4}}\sqrt{\Det(\mathbf{J})}>0,
	\end{aligned}
	\label{eq:pattern_condition}
	\end{equation}
	then $(u^*,v^*)$ is locally asymptotically stable with respect to the corresponding ODE system, but is unstable with respect to system (\ref{eq:subsystem_approx}).
\end{theorem}
\begin{proof}
The solution to (\ref{eq:dispersion}) yields:
	\begin{align}
	\alpha_{\pm}(k^2)=\frac{b(k^2)\pm \sqrt{[b(k^2)]^2-4c(k^2)}}{2}.
	\label{eq:solutiondispersion}
	\end{align}
Note that $k^2=0$ corresponds to the spatially homogeneous case. For Turing instability, we require the spatially homogeneous state $(u^*,v^*)$ to be stable in the absence of spatial variation ($k^2=0$). Therefore, for $k^2=0$, the eigenvalues given in (\ref{eq:solutiondispersion}) have to be negative. This occurs when the trace of $\mathbf{J}$ is negative and the determinant of $\mathbf{J}$ is positive. From the special form of our matrix $\mathbf{J}$ in (\ref{eq:cooperative_jacobian}), it is clear that the trace of $\mathbf{J}$ is always negative, and from the first assumption in \eqref{eq:pattern_condition} know that the determinant of $\mathbf{J}$ is positive. So $(u^*,v^*)$ is locally asymptotically stable with respect to the corresponding ODE system.

For the emergence of a non-constant spatially periodic solution, we further require that for some $k^2 \ne 0$, Re$(\alpha_+(k^2))>0$, guaranteeing that the perturbation will grow. Since $\Tr{\mathbf{D}}>0$ and $\Tr{\mathbf{J}}<0$, a necessary but not sufficient condition is that
   \begin{equation*}
   c(k^2)<0 \,\,\,\,\,\text{for some $k^2$} \in \R_{+}.
   \end{equation*}  
 This happens as long as the following condition is satisfied:
    \begin{align}
   d_{11}g_{v}+d_{22}f_{u}-f_{v}d_{21}>0.
   \label{eq:turing_condition1}
\end{align}
Under the condition \eqref{eq:turing_condition1}, the minimum of $c(k^2)$ is achieved at some $k^2=k_m^2>0$. Minimizing $c(k^2)$ with respect to $k^2$ yields:
\begin{equation}
\begin{aligned}
  	c_{min}=\min_{k^2}c(k^2)=c(k_m)=\Det(\mathbf{J})-\frac{(  d_{11}g_{v}+d_{22}f_{u}-f_{v}d_{21})^2}{4(\Det(\mathbf{D})-f_{v}h_{1})}, \,\, k_{m}^2=\frac{  d_{11}g_{v}+d_{22}f_{u}-f_{v}d_{21}}{2(\Det(\mathbf{D})-f_{v}h_{1})}.
  	\end{aligned}
  	\label{eq:minimum}
  	\end{equation}
  	Guaranteeing that $c_{min}<0$, we then have the following final condition:
  	\begin{equation}
   	\begin{aligned}
 d_{11}g_{v}+d_{22}f_{u}-f_{v}d_{21}>2\sqrt{\Det{\mathbf{D}}-f_{v}h_{1}}\sqrt{\Det{\mathbf{J}}}.
  	\end{aligned}
  	\label{eq:turing_conditions}
  	\end{equation}
\end{proof}

Now, the same range of wavenumbers $k$ that makes $c(k^2)<0$ in (\ref{eq:determ}) also guarantees that Re$(\alpha(k^2))>0$. We can further calculate the relevant range of wavenumbers $k_{-}^{2}<k^2<k_{+}^2$ by computing the zeros of the function $c(k^2)$ such that $c(k_-^2)=c(k_+^2)=0$. Then
\begin{equation}
\begin{aligned}
k_{-}^2=\frac{B(J,D)-\sqrt{B(J,D)^2-4\Det{J}(\Det{D}-f_{v}h_{1})}}{2(\Det{D}-f_{v}h_{1})}<k^2\\
<\frac{B(J,D)+\sqrt{B(J,D)^2-4\Det{J}(\Det{D}-f_{v}h_{1})}}{2(\Det{D}-f_{v}h_{1})}=k_{+}^2, 
\end{aligned}
\label{eq:interval}
\end{equation}
where
\begin{equation}
\begin{aligned}
B(J,D)=d_{11}g_{v}+d_{22}f_{u}-f_{v}\lambda v^*P_{2}
\end{aligned}
\end{equation}
The spatial patterns that emerge have a corresponding wavelength $\omega$, defined as 
\begin{equation*}
\begin{aligned}
\omega=\frac{2\pi}{k_{m}},
\end{aligned}
  	\end{equation*}
with $k_{m}$ defined in \eqref{eq:minimum} in the interval (\ref{eq:interval}) and the one for which the positive eigenvalue $\alpha_+(k^2)$ from (\ref{eq:solutiondispersion}) achieves a maximum, corresponding to the most unstable and fastest growing mode. 

\begin{figure}
\centering
\includegraphics[trim=0.9cm 0cm 0.9cm 0cm,width=0.98\textwidth,clip]{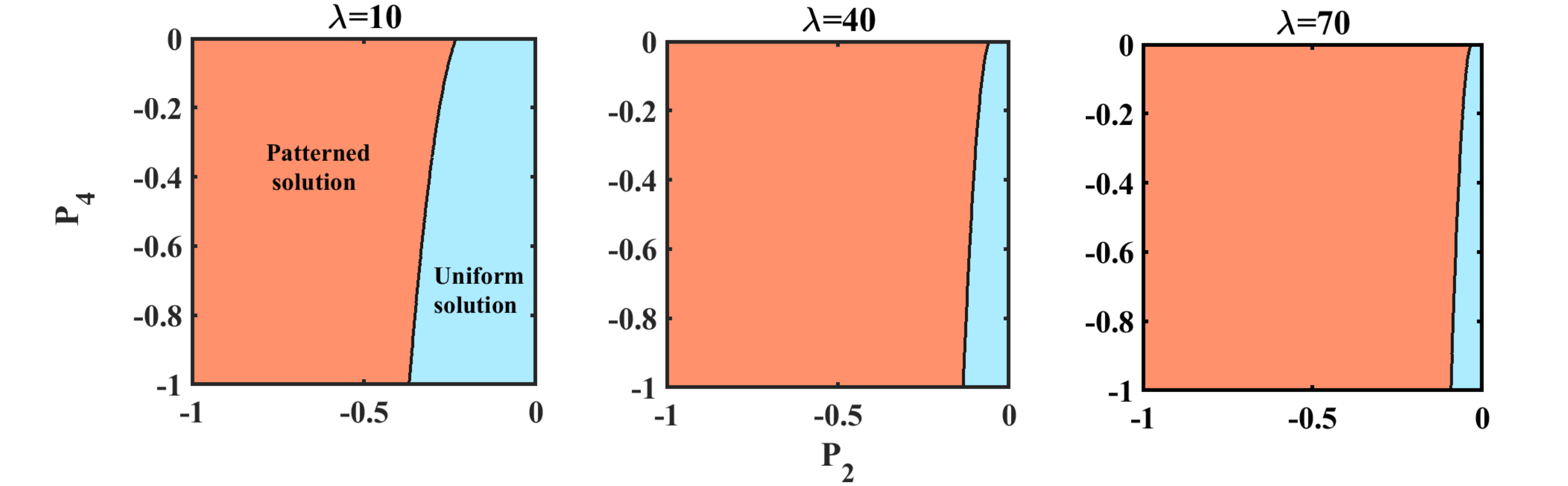}
\caption{Parameter space of Turing-like instability satisfying conditions (\ref{eq:turing_condition_space}) for various values of $P_{2}$ and $P_{4}$ with fixed $\lambda$. Note that increasing the value of $\lambda$, results in a larger parameter space. It is also clear that patterns are possible even in the absence of the fourth order term ($P_{4}=0$). The following parameters are used: $N=0.06,M=-0.13, \det{\mathbf{J}}= 0.004, \det{\mathbf{D}=0.024}$ derived from the original model with $e_{1}=1.05,p_{1}=3.5,\phi=0.14,e_{3}=0.1140,\delta=1/7,D_{G}=0.04,D_{S}=0.6$.}
\label{fig:6}
\end{figure} 

We can gain further insight into the result from Theorem \ref{theorem3} by visualizing the instability conditions in the $P_{2}P_{4}$-plane (Figure \ref{fig:6}). Letting
	\begin{equation}
	\begin{aligned}
	M=d_{11}g_{v}+d_{22}f_{u},\;\;  N&=f_{v}v^*,
	\end{aligned}
	\label{eq:M}
	\end{equation} 
the second stability condition from Theorem \ref{theorem3}  is equivalent to
\begin{equation}
\begin{aligned}
(M-\lambda NP_{2})^2-4\Det(\mathbf{J})(\Det(\mathbf{D})-\lambda NP_{4})>0.
\end{aligned}
\label{eq:turing_space}
\end{equation}
Further rearrangement of (\ref{eq:turing_space}) leads to the following condition for the instability of the uniform solution:
	\begin{equation}
	\begin{aligned}
P_{4}>\frac{-P_{2}^2N^2\lambda^2+2P_{2}MN\lambda-(M^2-4\Det(\mathbf{J})\Det(\mathbf{D}))}{4N\Det(\mathbf{J})\lambda}.
\end{aligned}
	\label{eq:turing_condition_space}
	\end{equation}
Additionally, we have
	\begin{equation}
   \begin{aligned}
	M^2-4\Det(\mathbf{D})\Det(\mathbf{J})=(d_{11}g_{v}-d_{22}f_{u})^2+4d_{11}d_{22}f_{v}g_{u}>0.
	\end{aligned}
	\label{eq:parabola_behavior}
	\end{equation}
Note that the first stability condition from \eqref{eq:pattern_condition} in Theorem \ref{theorem3} is independent of $P_{2},P_{4}$ and $\lambda$. Therefore, given that this first condition holds, rearranging the other instability condition in Theorem \ref{theorem3} in the form of (\ref{eq:turing_condition_space})  and using (\ref{eq:parabola_behavior}) as well as the fact that $M<0$ and $N>0$, it is clear that the instability region corresponds to the area to the left of the downward facing parabola defined on the right hand side of (\ref{eq:turing_condition_space}) in the fourth quadrant of the $P_{2}P_{4}$-plane (Figure \ref{fig:6}). Without the activator-inhibitor dynamics, a cooperative system cannot be destabilized by diffusion alone. Theorem \ref{theorem3} makes it clear that the additional cross-diffusion terms given by $P_{2}$ and $P_{4}$ make spatial heterogeneous patterns possible in this cooperative system. In particular, the cross-diffusion term with $P_{2}$ plays a crucial role in the pattern forming mechanism since in its absence ($P_{2}=0$),  the term $c(k^2)$ from the dispersion relation in (\ref{eq:dispersion}) can never be negative for any $k^2$. Since the biharmonic parameter $P_{4}$ acts as a stabilizing force, we also note that as its absolute value increases, the window for spatial patterns decreases (Figure \ref{fig:6}). Increasing the value of the strength parameter $\lambda$ offsets the effect of the biharmonic parameter $P_{4}$ and increases the size of the window in which spatial patterns are possible. We note that in the absence of the biharmonic long-range cross-diffusion term ($P_{4}=0$), the conditions for Turing instability can still be satisfied. In this case, the system is reduced to a special case of the reaction-diffusion model with cross-diffusion, for which Turing instability conditions have been previously derived \citep{madzvamuse2015cross}. 

Previous results in this section took into consideration the system on an infinite domain $\mathbb{R}$. In such a system, we will always find an unstable mode $k^2$ in the interval (\ref{eq:interval}) if the conditions in Theorem \ref{theorem3} are satisfied. Numerical simulations require the choice of a finite domain with specific boundary conditions. Therefore, we now consider the scenario on a bounded domain $T=(-l,l)$ with periodic boundary conditions, where the size of $T$ also affects the pattern formation. This is a more restrictive situation than the infinite domain scenario as the wavenumbers $k$ are now discrete and depend on the size of the domain. In this case, we shall understand that the solution $(u,v)$ on $T$ are periodically extended to $\mathbb{R}$ so the integral terms in the original system is still integrated on $\mathbb{R}$. 

The result for the bounded domain case is summarized in the following corollary:
\begin{corollary}\label{periodic_boundary}
	Consider (\ref{eq:subsystem_approx}) on a finite domain $T=\{x\in \mathbb{R}: -l < x<l\}$ and the following periodic boundary conditions:
	\begin{equation}\label{BC}
	\begin{aligned}
	u(-l,t)&=u(l,t),u_{x}(-l,t)=u_{x}(l,t),\\
	v(-l,t)&=v(l,t),v_{x}(-l,t)=v_{x}(l,t).
	\end{aligned}
	\end{equation}
	Let $(u^*,v^*)$ be a constant steady state solution of (\ref{eq:subsystem_approx}) with $\mathbf{D}$ defined as in (\ref{eq:assumptions}) with $d_{11},d_{22}>0$ and $P_{2}<0$ and $P_{4}<0$ defined in (\ref{kernel_definitions}). Also, let $\mathbf{J}$ be defined as in (\ref{eq:cooperative_jacobian}) with $f_{u},g_{v}<0$ and $f_{v},g_{u}>0$. If
	\begin{equation}
	\begin{aligned}
	\Det{\mathbf{J}}&>0,\\
		d_{11}g_{v}+d_{22}f_{u}-f_{v}\lambda v^*P_{2}-2\sqrt{\Det{\mathbf{J}}(\Det{\mathbf{D}}-f_{v}\lambda v^{*}P_{4})} &> \Big(\frac{\pi}{l}\Big)^2(\Det{\mathbf{D}}-f_{v}\lambda v^{*}P_{4}),
	\end{aligned}
	\label{eq:pattern_condition_finite}
	\end{equation}
then $(u^*,v^*)$ is locally asymptotically stable with respect to the corresponding ODE system, but is unstable with respect to system (\ref{eq:subsystem_approx}) on $T$ with boundary condition \eqref{BC}. Moreover the most unstable mode is given by $n\in \mathbb{N}$ such that 
\begin{equation}
    \alpha(k_m^2)=\alpha_+\left(\frac{n^2\pi^2}{l^2}\right)=\ds\max_{i\in {\mathbb N}}\alpha_+\left(\frac{i^2\pi^2}{l^2}\right),
\end{equation} where $\alpha_+(k^2)$ is defined in \eqref{eq:solutiondispersion}, and the corresponding wavelength is $\omega=2\pi/k_m=2l/n$. 
\end{corollary}

\begin{proof}
The non-constant eigenfunctions that satisfy the corresponding  eigenvalue problem $\phi''+\lambda \phi=0$ on the domain $(-l,l)$ with periodic boundary conditions are of the following form
\begin{align*}
\phi_i(x)= a_{1}\sin\left(\frac{i\pi x}{l}\right)+a_{2}\cos\left(\frac{i\pi x}{l}\right),\,\;  i \in \mathbb{N},
\end{align*}
and the corresponding eigenvalues are $k_i=(i\pi/l)^2$ for $i \in \mathbb{N}$.  Now, when $0<k_{-}^2< k_i^2=(i \pi/l)^2<k_{+}^2$  for some $i\in$ $\mathbb{N}$, where $k_{-}$ and $k_{+}$ are defined in (\ref{eq:interval}), the eigenvalue $\alpha_+(k_i^2)$ defined in (\ref{eq:dispersion}) is positive for this $i$.

We then note that the discrete wavenumber $k$ increases by $\pi/l$ with each $i$. Therefore, to guarantee that we have at least one $k^2=(i \pi/l)^2$ in the interval given by $(k_{-}^2,k_{+}^2)$, it is sufficient that the length of the interval $(k_{-},k_{+})$ is larger than $\pi/l$ \citep{shi2011cross}. Using 
	\begin{equation*}
	\begin{aligned}
	(k_{+}-k_{-})^2=(k_{-}^2+k_{+}^2)-2 k^+_{1}k^-_{2}> \Big(\frac{\pi}{l}\Big)^2,
	\end{aligned}
	\end{equation*}
		and the expressions of $k_{-}$ and $k_{+}$ in (\ref{eq:interval}), 
 we obtain the second  instability condition in \eqref{eq:pattern_condition_finite}.
Now, for an interval of length  $2l$, if the instability conditions \eqref{eq:pattern_condition_finite} are satisfied, then a spatially patterned solution will emerge with the corresponding wavenumber $k=i\pi/l$ such that $k^2\in (k_{-}^2,k_{+}^2)$. The most unstable wavenumber $k_m$ is the one that maximizes $\alpha_+(k^2)$ in (\ref{eq:solutiondispersion}).
\end{proof}
	
  \begin{figure}
	\centering
	\includegraphics[width=0.85\textwidth]{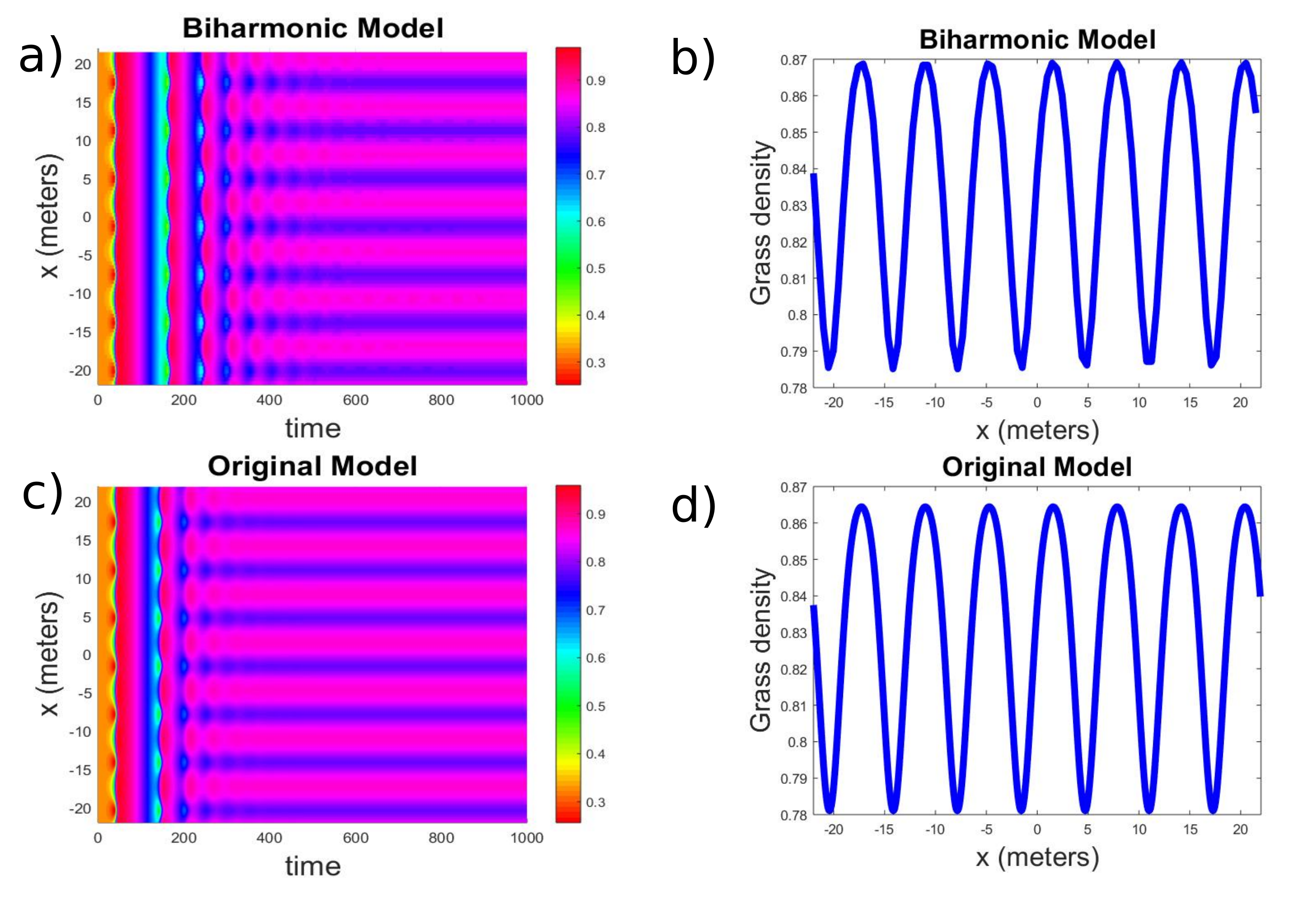}
	\caption{Spatial patterns produced through simulations of the biharmonic system (\ref{eq:original_biharmonic}) (panels a) and b) ) and original system (\ref{eq:subsystem_approx}) (panels c) and d) ) with Case I parameters: $D_{G}=0.04, D_{S}=0.6, \lambda=40,e_{1}=0.7,p_{1}=3.5, f=0.14,e_{3}=0.1140,\delta=\frac{1}{7}$ corresponding to the steady state values $G^*=0.33, S^*=2.75$. For the scale-dependent parameters, we use $\sigma_{1}=0.43$ and $\sigma_{2}=.68$ in the original model and corresponding values of $P_{2}=-0.1388$ and $P_{4}=-0.0225$ for the biharmonic model. Both simulations are performed on a bounded domain $T=(-l, l)=(-7\pi,7\pi)$. All parameters are chosen to satisfy conditions from (\ref{eq:pattern_condition_specific}). Panels a) and c) show temporal evolution of the grass density while panels b) and d) show the final steady state of grass after 1000 time units. The characteristic wavelength is accurately predicted as $\omega=\frac{14\pi}{7}$.}
	\label{fig:7}
\end{figure}

The second instability condition in \eqref{eq:pattern_condition_finite} also defines a minimal length $l_m$ for the emergence of the spatial patterns:
\begin{equation*}
    l>l_m=\pi \sqrt{\frac{\Det{\mathbf{D}}-f_{v}\lambda v^{*}P_{4}}{d_{11}g_{v}+d_{22}f_{u}-f_{v}\lambda v^*P_{2}-2\sqrt{\Det{\mathbf{J}}(\Det{\mathbf{D}}-f_{v}\lambda v^{*}P_{4})}}}.
\end{equation*}
This implies that in numerical simulations, if one chooses $l<l_m$, then no spatial patterns can be observed. On the other hand, when the length $l$ is large, then the interval $(k_{-},k_{+})$ may contain multiple  unstable wavenumbers $k=i\pi/l$, and the spatial patterns with all these wavenumebrs are possible but the one with most unstable wavenumber $k_m$ is the one most likely to be observed. 

\subsection{Grass-Sediment Cooperative System with Nonlocal Interactions}

 We  now apply these results to our Grass-Sediment system (\ref{eq:2}). The biharmonic approximation yields the following approximated system:
\begin{equation}
\left \{ 
 \begin{aligned}
\partial_t G&=D_{G}\partial^2_x G+G\Big(F(S)-G\Big), &x\in \mathbb{R}, t>0\\
\partial_t S&=D_{S}\partial^2_x S+\phi(-L(G) S+1)+\lambda S(P_{2}\partial^2_x G+P_{4}\partial^4_x G),&x\in \mathbb{R}, t>0\\
G(x,0)&=G_{0}(x,0)\ge 0,\,\,\, S(x,0)=S_{0}(x,0)\ge 0, &x\in \mathbb{R},
\end{aligned}\right.
\label{eq:original_biharmonic}
\end{equation}
with $F(S),L(G),P_{2}$ and $P_{4}$ defined previously in Sections 2 and 3.2. From Section 3.1, at the stable uniform positive steady state $(G_{+}^{*},S_{+}^*)$ defined in (\ref{eq:steady}), we have:
\begin{equation*}
\begin{aligned}
f_{u}&=-G_{+}^*,\,\,\, f_{v}=G_{+}^* F'(\frac{1}{L(G_{+}^*)})   ,\\
g_{u}&=-\phi \frac{1}{L(G_{+}^*)}L'(G_{+}^*),\,\,\, g_{v}= -\phi L(G_{+}^*).
\end{aligned}
\end{equation*}
Note that the cooperative form of this system with $f_{v},g_{u}>0$ and $f_{u},g_{v}<0$.
For numerical simulations, we consider this system on a finite domain $T=(-l,l)$ and the following periodic boundary conditions:
\begin{equation*}
\begin{aligned}
G(-l,t)&=G(l,t),\;\;G_{x}(-l,t)=G_{x}(l,t),\\
S(-l,t)&=S(l,t),\;\;S_{x}(-l,t)=S_{x}(l,t).
\end{aligned}
\end{equation*}
Using the results from Theorem \ref{theorem3} and Corollary \ref{periodic_boundary}, we have the following condition necessary for Turing instability on $T$:
\begin{equation}
\begin{aligned}
D_{G}g_{v}+D_{S}f_{u}-f_{v}\lambda S^*P_{2}-2\sqrt{\Det(\mathbf{D})-f_{v}\lambda S^*P_{4}}\sqrt{\Det(\mathbf{J})}&>\Big(\frac{\pi}{l}\Big)^2(\Det{\mathbf{D}}-f_{v}\lambda S^*P_{4}),
\end{aligned}
\label{eq:pattern_condition_specific}
\end{equation}
where
\begin{equation*}
\mathbf{D}=\left( \begin{array}{cc}
D_{G}	&  0\\ 
\lambda S^* P_{2}	& D_{S} \end{array} \right),\, 
\mathbf{J} =	\left( \begin{array}{cc}
f_{u}	&  f_{v}\\ 
g_{u}& g_{v}
\end{array} \right) \Bigg\rvert_{(G^*,S^*)}.
\label{eq:assumptions2}
\end{equation*}

\begin{figure}[t]
	\centering
	\includegraphics[trim=1.5cm 0.1cm 2.3cm 0cm,width=0.9\textwidth,clip]{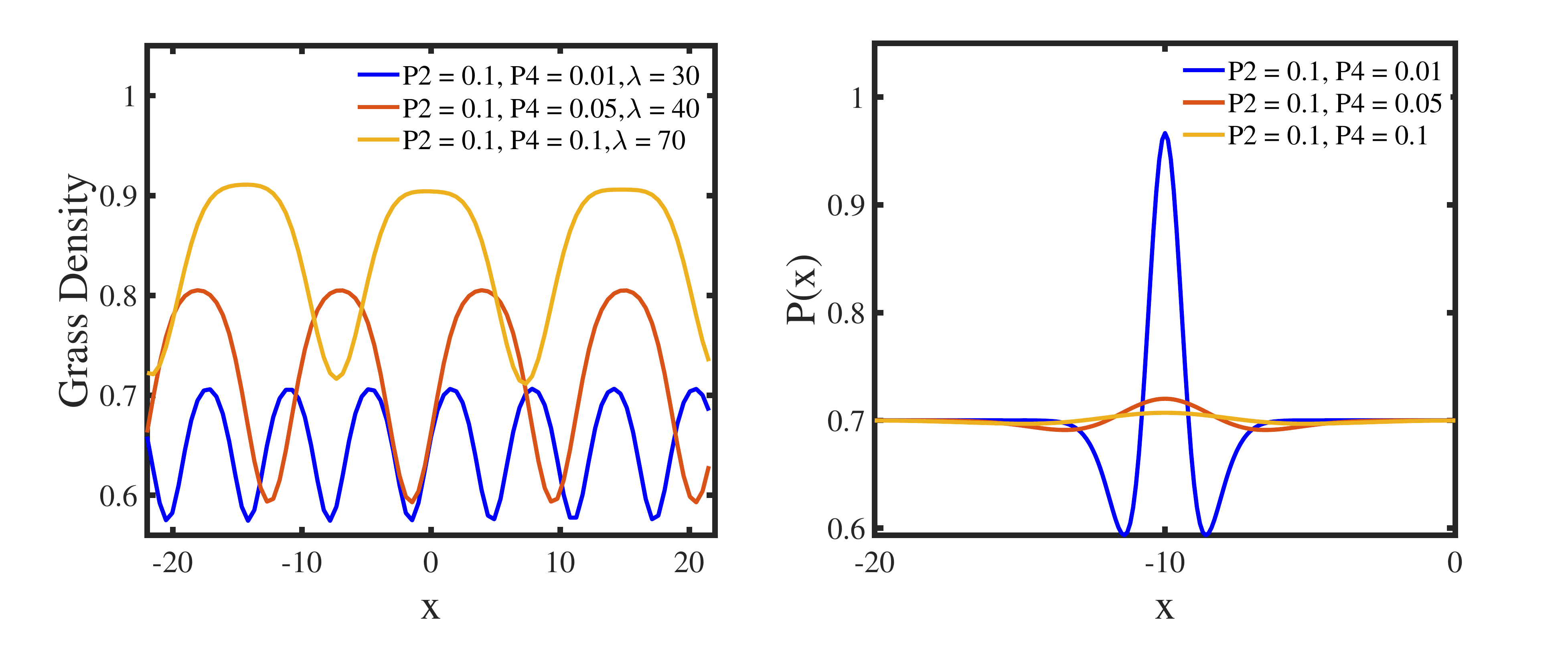}
	\caption{Panel a) displays the final steady states of grass after 1800 time units in the biharmonic system (\ref{eq:original_biharmonic}) for various values of $P_{2}$, $P_{4}$ and $\lambda$. Note that the value of $\lambda$ has to be adjusted to offset increasing $P_{4}$ in order for patterns to emerge. Biologically realistic parameters are chosen such that conditions (\ref{eq:pattern_condition_specific}) are satisfied: $e_{1}=1.05,p_{1}=3.5,\phi=0.14,e_{3}=0.1140,\delta=1/7,D_{G}=0.04,D_{S}=0.6$. Panel b) displays plots of kernel functions in (\ref{kernel}) corresponding to parameters $P_{2}$ and $P_{4}$ from panel a) with larger values of parameter $P_{4}$ resulting in wider kernels.} 
	\label{fig:8}
\end{figure}

For Case I parameter regime from Section 3.1, we choose:  $e_{1}=0.7,p_{1}=3.5, f=0.14,e_{3}=0.1140,\delta=1/7$. We then choose $D_{G}=0.04, D_{S}=0.6, \lambda=40$ for our nonlocal parameter values to satisfy the instability conditions (\ref{eq:pattern_condition_specific}) and numerically integrate the biharmonic system (\ref{eq:original_biharmonic}) on $T=(-l,l)=(-7\pi,7\pi)$. We find that a spatially patterned solution emerges, as predicted (Figure \ref{fig:7}) and these simulations are also consistent with numerical simulations of the original system (\ref{eq:2}), suggesting that the theoretical results derived from the biharmonic system can be applied to the original system to give insight regarding under what conditions a spatially patterned solution emerges. The eigenvalue $\alpha_+(k^2)$ is given by:
\begin{equation}
\begin{aligned}
\alpha_+(k^2)=-0.1893-0.32k^2+\frac{\sqrt{(-0.37858-0.64k^2)^2-0.44393k^4+1.35564k^2-0.03963}}{2}.
\end{aligned}
\label{eq:eigenvalue_specific}
\end{equation}
Furthermore, on the domain $T=(-7\pi,7\pi)$, the range of wavenumbers for which the corresponding eigenvalue $\alpha_+(k^2)$ is positive is given by:
\begin{equation}\label{pp}
\begin{aligned}
k_{-}^2=0.0295< k^2=\left(\frac{i\pi}{7\pi}\right)^2 <3.0242=k_{+}^2, \,\,\, i \in \mathbb{N}.
\end{aligned}
\end{equation}
It can be calculated that for $2\le i\le 12$, \eqref{pp} is satisfied, and when $i=7$, $\alpha_+(k^2)$ is maximized. Hence
the characteristic wavelength of the emerging patterns is $\omega=2l/7=2\pi$. This is consistent with simulation results which show $7$ peaks (Figure \ref{fig:7}). Similar results are obtained for Case II parameter regime (see Figure \ref{fig:11} in the Appendix.)

 \begin{figure}
	\centering
	\includegraphics[width=0.55\textwidth]{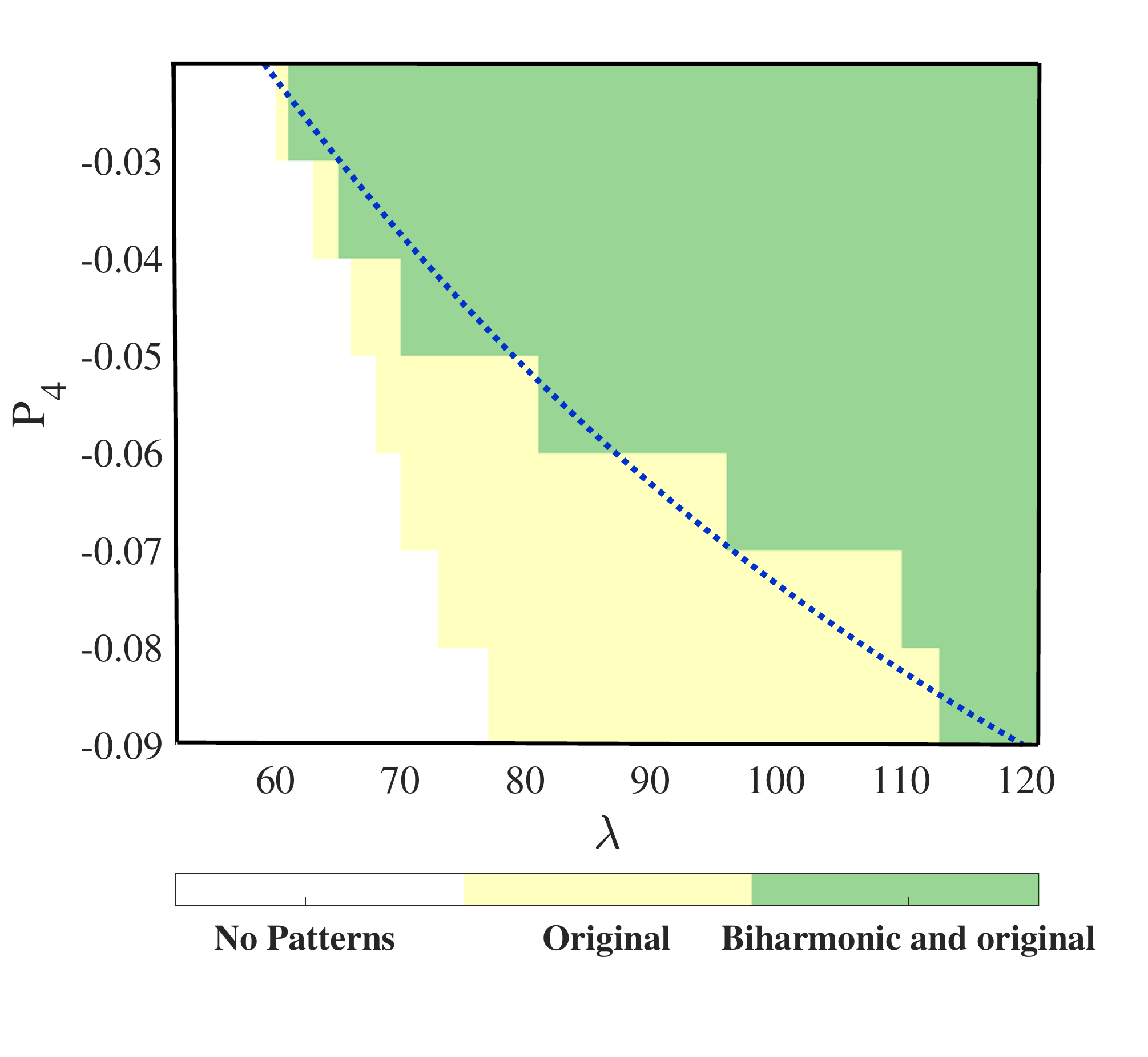}
	\caption{We numerically integrate the fourth order biharmonic system (\ref{eq:original_biharmonic}) for different kernel widths $(P_{4})$ and kernel strengths $(\lambda)$. Similarly, simulations of the original system (\ref{eq:2}) are performed for the corresponding kernel parameters $\sigma_{1}$ and $\sigma_{2}$ calculated from (\ref{kernel_definitions}). All simulations are run over the same domain $T=(-7\pi,7\pi)$ with the following parameters: $D_{G}=0.04,D_{S}=0.6,e_{1}=1.05,p_{1}=3.5,\phi=0.14,e_{3}=0.1140,\delta=1/7,P_{2}=-0.05$. The region of instability derived from Corollary \ref{periodic_boundary} corresponds to the area above the dotted blue curve. The results shown in green correspond to the instances where spatial patterns emerge for both the biharmonic system (\ref{eq:original_biharmonic}) and the original system (\ref{eq:2}), while the results in yellow correspond to instances where patterns emerged only for the original system (\ref{eq:2}). It is clear that the theoretical results from  Corollary \ref{periodic_boundary} (light blue curve) are consistent with the numerical simulations of the biharmonic system (region in green) and less consistent with the original system (region in yellow). Although these results are not as consistent, it is clear that the theoretical results can nonetheless be used to predict the formation of patterns in the original system (\ref{eq:2}).}.
	\label{fig:9}
\end{figure}

Previously, we used biologically realistic parameters from Table 1 (see the Appendix) to perform all numerical simulations, including realistic parameters for the scale-dependent feedback ($P_{2}$, $P_{4}$, $\lambda$). Now, we are interested in how varying these scale-dependent feedback parameters may affect the nature of the spatial patterns in system (\ref{eq:original_biharmonic}). As predicted in Section 3.2, since the biharmonic term $P_{4}$ acts as a stabilizing force, as its value gets larger, the window for spatial patterns decreases and a larger value of $\lambda$ is necessary to offset its effect and allow spatial patterns to emerge (Figure \ref{fig:8}a). In addition, choosing a larger value of $P_{4}$ results in an overall increase in the pattern wavelength (Figure \ref{fig:8}a). This result can also be interpreted in the context of how the coefficients $P_{2}$ and $P_{4}$ are related to the shape of the kernel in (\ref{kernel}) in the original system (\ref{eq:2}) (Figure \ref{fig:8}b). The coefficient $P_{2}$ measures the difference of the variances $\sigma_{1}$ and $\sigma_{2}$ of the excitatory and inhibitory interactions, respectively. The coefficient $P_{4}$ is related to kurtosis and controls the weight of the kernel's tails while $\lambda$ modulates the amplitude of the Mexican-hat kernel. Since for a fixed value of $P_{2}$, an increase in $P_{4}$ results in a wider, flatter kernel shape, the wider the range of the long-range effects given by $P_{4}$, the stronger these interactions need to be (given by $\lambda$) to have a significant effect and lead to the formation of spatial patterns. This makes biological sense, since the intensity of scale-dependent interactions tend to dissipate over larger distances and therefore need to be amplified to have any effect on spatial heterogeneity over longer ranges. In addition, we see that wider kernels result in wider spatial patterns characterized by longer wavelengths. Again, this makes biological sense as one would expect the scale of the spatial interactions to influence the resulting spatial patterns.

Finally, we compare our analytic results with numerical simulations of the approximated biharmonic system (\ref{eq:original_biharmonic}) and the original system (\ref{eq:2}) (Figure \ref{fig:9}). The analytic results from (\ref{eq:pattern_condition_specific}) are consistent with numerical simulations of the biharmonic system and the original system. However, we note that the onset of patterns in the original system occurs sooner than in the biharmonic system. Nonetheless, this result suggests that using the biharmonic system can help find the relevant parameter regime in which spatial patterns are possible in the original system and gain understanding into how the nature of the scale-dependent feedbacks affects the development of spatial patterns.

\section{Discussion} 
\label{sec:8}

We propose a phenomenological model to describe the dynamics of the marsh edge in terms of two-way interactions between marsh grass \textit{Spartina alternifora} and sedimentation. In nature, the marsh edge can frequently be observed in a number of configurations ranging from a spatially uniform  to a more wave-like shoreline. The interest of this paper lies in understanding whether the well-known scale-dependent (nonlocal) feedback between marsh vegetation and sedimentation can lead to spatially variable shoreline configurations. Marsh grass promotes sediment accretion in its immediate surroundings by slowing down current acts as a facilitation mechanism. In turn, the diverted water flow contributes to increased erosion further away and acts as an inhibitory mechanism. We propose a system of reaction-diffusion equations with an additional integral term with a Mexican-hat kernel function that describes the nature of this scale-dependent feedback. Our system is highly cooperative; as cooperative systems often lack the classic activator-inhibitor mechanism necessary for pattern formation, it becomes of interest how and under what conditions spatial patterns may develop. 

We perform a biharmonic approximation of our system and carry out analysis on the simpler biharmonic system that expresses the kernel function as separate short-range and long-range diffusion terms. Using the more mathematically tractable biharmonic system, we are then able to derive general condition for the formation of spatial patterns in a cooperative system such as ours. Further, using numerical simulations, we confirm that the biharmonic model, while an approximation, is consistent with the original model, and therefore we can apply the theoretical results from the biharmonic system to help gain insight into the formation of patterns in the original system. We parameterize the kernel function using a set of reasonable parameters from literature and find that spatial patterns can develop, given that the scale-dependent interactions between marsh vegetation and sediment dynamics are strong enough. The model thus provides further evidence that the presence of scale-dependent interactions is essential for pattern formation and that heterogeneous patterns cannot occur in the presence of weak scale-dependent interactions. Not surprisingly, we find that the choice of wider kernels tend to produce wider spatial patterns (characterized by longer wavelengths) and vice versa. The nature and strength of the grass-sediment scale-dependent interactions depends on many factors such as the underlying hydrodynamics and sediment composition, the exact spatial scale (corresponding to the widths of the Mexican-hat kernel) and relative strength of the scale-dependent feedback are difficult to estimate in the field and can vary substantially. We use one possible set of biologically realistic parameters for the kernel function (Table \ref{Table:Parameters}) and find that the patterns that emerge in simulations occur on a spatial scale consistent with what can observed in nature (4-10 meters between peaks) \citep{vandenbruwaene2011flow}. 

Furthermore, we find that there are two possible parameter regimes in the system. The first regime is especially of interest as it corresponds to a more realistic scenario where marsh vegetation is effective at attenuating erosion through the binding of sediment and decreasing the effect of wave erosion. Given the strong facilitatory nature of the grass-sediment interactions, bistability takes place in this parameter regime. In general, bistable dynamics makes a system especially prone to collapsing to an irreversible state as environmental conditions gradually worsen and a tipping point is reached \citep{dakos2011slowing, kefi2014early, kefi2016can} through the phenomenon of hysteresis. Pattern formation has previously been suggested as a possible coping mechanism for systems close to degradation \citep{chen2015patterned}. The analysis in this paper gives more insight into this previously reported phenomenon as we also find this to be the case in our model \citep{zaytseva2018} where pattern formation allows the marsh edge to cope with harsher erosion through spatial variation. One limitation of our model is the lack of multiple spatial dimensions as only the dynamics on a one-dimensional cross-section of the marsh edge were considered. Hence, we were not able to observe the geometry of the protrusions. In addition, the model is meant to be phenomenological in nature, omitting processes such as the effect and variation of hydrodynamics and wave action, modeled in more detail previously \citep{fagherazzi2012numerical}. 
Despite the relatively simple dynamics of our one-dimensional model, it is able to capture the pattern formation on the marsh edge as a result of scale-dependent feedbacks between vegetation and sediment accumulation. The agreement between the model simulations and field observations suggests that important pattern-generating
processes have been captured in the model and non-local interactions between plants and sedimentation can drive the formation of shoreline patterns. In addition, the results in this paper can be generalized to any cooperative system with scale-dependent feedbacks in the form of short-range activation and long-range inhibition, described using a Mexican-hat kernel function.

\vspace{0.1in}


\section{Appendix}
Figure \ref{fig:10} shows the plot of the functional forms of $F(S)$ and $L(G)$.
\begin{figure}
	\centering
	\includegraphics[trim=5cm 19cm 6.5cm 4cm, width=0.75\textwidth,clip]{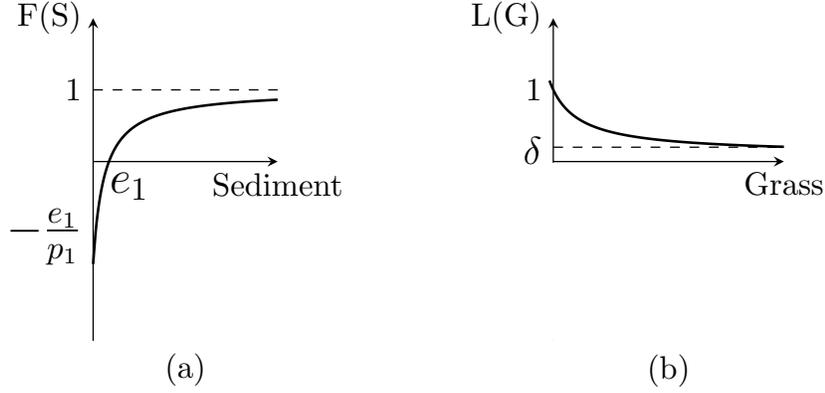}
	\caption{Functions a) $F(S)$ and b) $L(G)$ from (\ref{eq:scaled_functions})}.
	\label{fig:10} 
\end{figure}
Figure (\ref{fig:11}) shows numerical simulations of both the biharmonic system (\ref{eq:original_biharmonic}) and the original system (\ref{eq:2}) for the parameter regime in Case II from Section 3.1. We see that a spatially patterned solution emerges if the instability conditions in (\ref{eq:pattern_condition_specific}) are satisfied. 
  \begin{figure}
	\centering
	\includegraphics[width=0.85\textwidth]{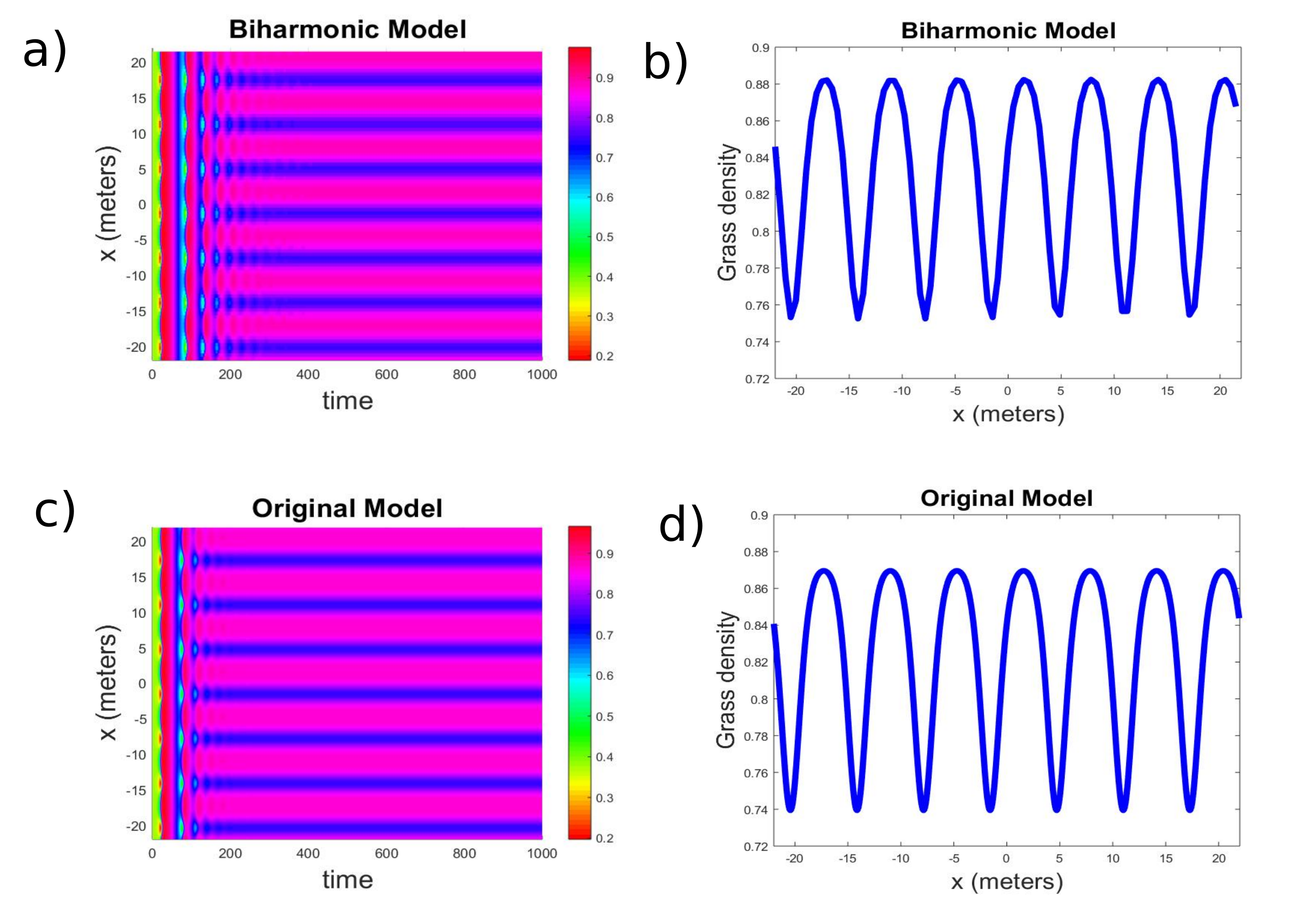}
	\caption{Spatial patterns produced through simulations of the biharmonic system (\ref{eq:original_biharmonic}) (panels a) and b) ) and original system (\ref{eq:subsystem_approx}) (panels c) and d) ) with Case II parameters:$D_{G}=0.04, D_{S}=0.6, \lambda=40,e_{1}=0.7,p_{1}=0.5, f=0.14,e_{3}=0.5,\delta=0.3$. For the scale-dependent parameters, we use $\sigma_{1}=0.43$ and $\sigma_{2}=.68$ in the original model and corresponding values of $P_{2}=-0.1388$ and $P_{4}=-0.0225$ for the biharmonic model. Both simulations are performed on a bounded domain $T=(-7\pi,7\pi)$. All parameters are chosen to satisfy conditions from (\ref{eq:pattern_condition_specific}). Panels a) and c) show temporal evolution of the grass density while panels b) and d) show the final steady state of grass after 1000 time units. The characteristic wavelength is accurately predicted as $\omega=\frac{14\pi}{7}$.}.
	\label{fig:11}
\end{figure}
Table \ref{Table:Parameters} shows the biologically realistic parameters for the original system and their sources. We use the parameter values from Table \ref{Table:Parameters} to obtain the new re-scaled parameters from Section 2 to use in all numerical simulations performed in this paper.
 \begin{table}[ht]
 	\caption{Biologically realistic parameters for the original system before non-dimensionalization}
 	\label{Table:Parameters}
 	\centering
 
      \resizebox{\textwidth}{!}{
 		\begin{tabular}{p{1.5cm}p{5.5cm}p{3cm}p{3cm}p{5.5cm}}
 			\hline
 			Symbol & Meaning & Unit & Value & Source \\ 
 			\hline
 			$\hat{D_{G}}$ &cordgrass diffusion coefficient& $m^2$ yr$^{-1}$&0.06 - 0.135&\citep{adams2012plant}\\ 
 			$\hat{D_{S}}$ &sediment diffusion coefficient& $m^2$ yr$^{-1}$&0.876&\citep{Liu14a}\\
 			$c$ &self-limiting growth rate of cordgrass& $m^2$ shoots$^{-1}$ yr$^{-1} $& 0.0057& \citep{Yang14}\\
 			$\psi$ &minimum erosion rate& yr$^{-1}$ &0.002-0.3&\citep{hardaway1999shoreline,rosen1980erosion}\\
 			$k_{s}$ &cordgrass density at which marsh erosion is half-maximal & shoots $m^{-2}$&30-50&estimated\\
 			$g$ &erosion constant in the absence of cordgrass & non-dimensional&5&\citep{mariotti2010numerical,sheehan2015tidal}\\
 			$\eta$ & sediment deposition rate& $m$ yr$^{-1}$& 0.002-0.006&\citep{stumpf1983process,goodman200717}\\
 			$p^*$ &intrinsic growth rate of cordgrass& yr$^{-1}$&1.5& \citep{Yang14}\\
 			$l_{1}$ &sediment threshold for cordgrass persistence& $m$ &0.02& estimated\\
 			$l_{1}^*$ &sediment elevation at which cordgrass growth is half-maximal& $m$ & 0.06& estimated\\
 			$\hat{\lambda}$ &strength of nonlocal cordgrass-sediment interactions& $m^2$ shoots$^{-1}$yr$^{-1}$&0.0004-0.3&\citep{bouma2007spatial} \\
 			$\sigma_{1}$ & standard deviation of the excitatory feedback for cordgrass&$m$&0.43& \citep{bouma2007spatial}\\
 			$\sigma_{2}$ & standard deviation of the inhibitory feedback for cordgrass&$m$&0.68&\citep{bouma2007spatial}\\
 			\hline
 		\end{tabular}}
  \end{table}
 
  All numerical simulations in this paper are performed using MATLAB. We use an implicit finite differencing scheme to numerically integrate the original equation. Although this scheme is more computationally intensive, it is chosen because it is always numerically stable and convergent. Because domain size plays an important role in the system's ability to form patterns, a large enough domain has to be chosen to be able to fit patterns with their characteristic wavelength. We evaluate all integrals using the \textit{trapz} function in MATLAB, which performs numerical integration using the trapezoidal rule. For the convolution term, we evaluate the integral of the product of the kernel and the periodically extended solution on the interval $(-3l,3l)$, to make sure an adequate number of kernels are considered in calculating the net effect. To numerically integrate the biharmonic system, we use an explicit finite differencing scheme in MATLAB. This scheme is less computationally intensive, and is easier to implement, given the extra biharmonic term. For both models, the numerical simulations are performed on a spatial domain $(-l,l)$ with $l=7\pi$ with periodic boundary conditions.  We apply Turing's idea of diffusion driven instability and use a spatially periodic perturbation of the stable steady state of the corresponding system of ODEs as the initial condition for our simulations.
\bibliographystyle{agsm}

\end{document}